%% file: remt-writeup.tex
\documentclass[11pt]{article}

\usepackage{sn-preamble} 
\usepackage{tikz}
\usepackage{verbatim}
\usepackage{cancel}
\usepackage{caption}
\usepackage{thm-restate}

\usepackage{microtype}

\usetikzlibrary{decorations.pathreplacing,angles,quotes}
\usetikzlibrary{shapes.geometric}
\usetikzlibrary{patterns}

\usepackage{mathtools} 
\usepackage{bm} 
\usepackage{mdframed}
\makeatletter
\newcommand{\setword}[2]{%
  \phantomsection
  #1\def\@currentlabel{\unexpanded{#1}}\label{#2}%
}
\makeatother

\newcommand{\ignore}[1]{}

\newcommand{\reldist}{\mathsf{rel}\text{-}\mathsf{dist}}
\newcommand{\ham}{\mathsf{ham}\text{-}\mathsf{dist}}

\newcommand{\Leb}{\mathsf{Leb}}

\newcommand{\cmon}{\mathcal{C}_{\textsf{monotone}}}

\newcommand{\Dyes}{\calD_{\mathrm{yes}}}
\newcommand{\Dno}{\calD_{\mathrm{no}}}

\newcommand{\yes}{\mathrm{yes}}
\newcommand{\no}{\mathrm{no}}

\newcommand{\SAMP}{\mathrm{Samp}}
\newcommand{\MQ}{\mathrm{MQ}}

\begin{document}

\title{Relative-error monotonicity testing}

\author{
Xi Chen \thanks{Columbia University.} \and 
Anindya De \thanks{University of Pennsylvania.} \and 
Yizhi Huang \thanks{Columbia University.} \and 
Yuhao Li \thanks{Columbia University.} \and 
Shivam Nadimpalli \thanks{Columbia University.} \and 
Rocco A. Servedio \thanks{Columbia University.} \and
Tianqi Yang \thanks{Columbia University.} 
\vspace{0.5em}
}

\date{
}

\pagenumbering{gobble}

\maketitle  

\begin{abstract}
The standard model of Boolean function property testing is not well suited for testing \emph{sparse} functions which have few satisfying assignments, since every such function is close (in the usual Hamming distance metric) to the constant-0 function.  In this work we propose and investigate a new model for property testing of Boolean functions, called \emph{relative-error testing}, which provides a natural framework for testing sparse functions.  

This new model defines the distance between two functions $f, g: \zo^n \to \zo$  to be $$\reldist(f,g) := { \frac{|f^{-1}(1) \hspace{0.06cm} \triangle \hspace{0.06cm} g^{-1}(1)|} {|f^{-1}(1)|}}.$$ This is a more demanding distance measure than the usual Hamming distance ${ {|f^{-1}(1) \hspace{0.06cm}\triangle \hspace{0.06cm} g^{-1}(1)|}/{2^n}}$ when $|f^{-1}(1)| \ll 2^n$; to compensate for this, algorithms in the new model have access both to a black-box oracle for the function $f$ being tested and to a source of independent uniform satisfying assignments of $f$.

In this paper we first give a few general results about the relative-error testing model; then, as our main technical contribution, we give a detailed study of algorithms and lower bounds for relative-error testing of \emph{monotone} Boolean functions. We give upper and lower bounds which are parameterized by $N=|f^{-1}(1)|$, the sparsity of the function $f$ being tested. 
Our results show that there are interesting differences between relative-error monotonicity testing of sparse Boolean functions, and monotonicity testing in the standard model. These results motivate further study of the testability of Boolean function properties in the relative-error model.
\end{abstract}

\newpage
\tableofcontents

\newpage
\setcounter{page}{1}
\pagenumbering{arabic}

\input{sections/intro}

\input{sections/prelims}

\input{sections/general-results}

\input{sections/upper-bounds}

\input{sections/two-sided-nonadaptive-lb}

\input{sections/two-sided-adaptive-lb}

\begin{flushleft}
\bibliographystyle{alpha}
\bibliography{allrefs}
\end{flushleft}

\appendix

\input{sections/ap-standard-yes-relative-no}
\input{sections/ap-mono-needs-both}

\end{document}

%% file: sections/intro.tex
\section{Introduction}
\label{sec:intro}

Over the past several decades, property testing of Boolean functions has blossomed into a rich field with close connections to a number of important topics in theoretical computer science including sublinear algorithms, learning theory, and the analysis of Boolean functions~\cite{BLR93, Rubinfeld:06survey, GGR98, Fischer:01, odonnell-book, Ron:08testlearn}.  The touchstone problems in Boolean function property testing, such as monotonicity testing and junta testing, have been intensively studied in a range of different models beyond the ``standard'' Boolean function property testing model which is described below; these include distribution-free testing, tolerant testing, active testing, and sample-based testing~\cite{HalevyKushilevitz:07,  parnas2006tolerant, BBBY12, KR00}. 

The present work makes two main contributions: At a conceptual level, we introduce and advocate a natural new model for property testing of Boolean functions extending the standard model, which we call \emph{relative-error property testing}.  At a technical level, we give a detailed study of the problem of \emph{monotonicity testing} in this new relative-error model.

\medskip

\noindent {\bf Motivation.}
In the standard model of Boolean function property testing, a testing algorithm gets black-box access to an arbitrary unknown function $f: \zo^n \to \zo$ and its goal is to tell whether $f$ has the property or has distance at least $\eps$ from every function satisfying the property, where the distance between two Boolean functions $f,g$ is taken to be the Hamming distance 
\[
\ham(f,g) := {\frac {|f^{-1}(1) \ \triangle \ g^{-1}(1)|}{2^n}},
\]
i.e.~the fraction of points in $\zo^n$ on which they disagree.  This model is closely analogous to the standard ``dense graph'' property testing model \cite{GGR98}, in which the testing algorithm may make black-box queries for entries of the unknown graph's adjacency matrix and the distance between two undirected $n$-node graphs is the fraction of all ${n \choose 2}$ possible edges that are present in one graph but not the other.

As is well known, the standard dense graph property testing model is not well suited to testing \emph{sparse} graphs, since under the Hamming distance measure mentioned above every sparse graph with $o(n^2)$ edges has distance $o(1)$ from the $n$-node empty graph with no edges.  Consequently, alternative models were developed, with different distance measures, for testing sparse graphs; these include  the bounded-degree graph model (see \cite{GoldreichRon02} and Chapter~9 in \cite{PropertyTestingICS}) and the general graph model (see \cite{ParnasRon02} and Chapter~10 in \cite{PropertyTestingICS}). 

What if we are interested in testing \emph{sparse} Boolean functions, i.e.~functions $f: \zo^n \to \zo$ that have $f^{-1}(1) \ll 2^n$?  Analogous to the situation for graphs that was described above, every sparse Boolean function is Hamming-distance-close to the constant-0 function, so the standard Boolean function property testing model is not well suited for testing sparse functions.  

\medskip

\noindent {\bf Relative-error property testing.}  We propose a new framework, which we call \emph{relative-error property testing}, which is well suited for testing sparse Boolean functions. The relative-error property testing model differs from the standard  model in the following ways:

\begin{flushleft}\begin{itemize}
\item 
First, we define the \emph{relative distance}\footnote{See the beginning of \Cref{sec:prelims} for a discussion of some of the basic properties of this definition.} from a function $f: \{0,1\}^n \to \{0,1\}$ to a function $g: \{0,1\}^n \to \{0,1\}$ to be 
\[
\reldist(f,g) := {\frac {|f^{-1}(1) \ \triangle \ g^{-1}(1)|}{|f^{-1}(1)|}}, \quad \text{i.e.}
\quad
\reldist(f,g) = \ham(f,g) \cdot {\frac {2^n}{|f^{-1}(1)|}}.
\]
Relative-error testing uses relative distance, rather than Hamming distance, to measure distances between functions.

\item Second, in the relative-error testing model, the testing algorithm has access to two different oracles:  (i) a \emph{black-box} (also called \emph{membership query}) oracle $\MQ(f)$ as in the standard model, which is queried with a point $x$ and returns $f(x) \in \{0,1\}$; and also (ii) a \emph{sample oracle} $\SAMP(f)$, which takes no input and, when queried, returns a uniform random $\bx \sim f^{-1}(1)$.  \end{itemize}
\end{flushleft}

\noindent {\bf Motivation and rationale for the relative-error testing model.}
The motivation for a testing model which can handle sparse functions is clear:  just as there are many interesting graphs which are sparse, there are many interesting Boolean functions which have relatively few satisfying assignments.  (This is especially evident when we view a Boolean function as a classification rule which outputs 1 on positive examples of some phenomenon which may be elusive; for example, generic inputs might correspond to generic candidate drugs or molecules, while positive examples correspond to candidates which have some rare but sought-after characteristic.)
The distance measure $\reldist(f,g)$ is natural to use when studying sparse Boolean functions, since it captures the difference between $f$ and $g$ ``at the scale'' of the sparse function $f$, whatever that scale may be. 

We observe that while the relative distance between two Boolean functions is not symmetric (i.e. $\reldist(f,g)$ is not necessarily equal to $\reldist(g,f)$), this does not pose a problem for us.  The reason is that we are interested in the setting where $\reldist(f,g)$ is small, and it is easily verified that if $\reldist(f,g) = \eps$ is at most (say) $0.99$, then $\reldist(g,f)$ is also $O(\eps)$; so in the regime of interest to us, relative distance is symmetric ``up to constant factors.''  We further observe that in the regime of our interest,  the definition of relative distance is quite robust:  as long as $\reldist(f,g)$ is at most a small constant, then replacing $|f^{-1}(1)|$ in the denominator by any of $|g^{-1}(1)|$, $|f^{1}(1)| + |g^{-1}(1)|$,  $\max\{|f^{-1}(1)|,|g^{-1}(1)|\}$ or  $\min\{|f^{-1}(1)|,|g^{-1}(1)|\}$ only changes the value of $\reldist(f,g)$ by a small constant factor.

The presence of a $\SAMP(f)$ oracle in the relative-distance model makes it possible to have non-trivial testing algorithms for sparse functions, since if only a black-box oracle were available then a huge number of queries could be required in order to find any input on which the unknown sparse function outputs 1.  (Indeed, it is difficult to imagine a reasonable testing scenario in which the testing algorithm does not have some kind of access to positive examples of the function being tested.)

\medskip

\noindent {\bf Relation to the standard model.}
Simple arguments which we give in \Cref{sec:general-results} show that if a property is efficiently relative-error testable then it is also efficiently testable in the standard model, and moreover that there are properties which are efficiently testable in the standard model but not in the relative-error model.  
Thus, relative-error testing is at least as hard as standard-model testing.
 A wide range of natural questions about the testability of well-studied Boolean function properties present themselves for this model; we discuss some of these questions at the end of this introduction.  The main technical results of this paper, though, deal with relative-error testing of \emph{monotonicity}, which is one of the most thoroughly studied Boolean function properties~\cite{GGLRS,FLNRRS,CS13a, CS13b,CST14, CDST15, KMS18}.  (Recall that a Boolean function $f: \zo^n \to \zo$ is \emph{monotone} if whenever $x,y \in \zo^n$ have $x_i \leq y_i$ for all $i$, i.e.~$x \preceq y$, it holds that $f(x) \leq f(y).$)

\subsection{Our results}

{\bf Comparison with standard-error model:} We start in \Cref{sec:general-results} by noting some basic points of comparison between the relative-error property testing model vis-a-vis the standard-error property testing~model for general properties of Boolean functions. First of all, it is not too hard to establish (\Cref{obs:relative-yields-standard}) that any property ${\cal C}$ which is testable in the relative-error model is also testable in the standard-error model. 
 In the other direction,  any algorithm to test a property ${\cal C}$ in the standard error model can be used to test ${\cal C}$ in the relative error model (see \Cref{obs:relative-from-standard}). However, the overhead of this simulation grows as $1/p$, where $p:=|f^{-1}(1)|/2^n$ is the density of the target function $f$. Thus the most interesting and challenging regime for relative error property testing is when the target function is sparse --- i.e., $p$ is vanishing as a function of $n$ ($p$ may even be as small as $2^{-\Theta(n)}$). 

\medskip
\noindent
 {\bf Upper and lower bounds for monotonicity testing:} We focus on 
  monotonicity testing in the relative-error model. 
Using the analysis of the ``edge tester'' \cite{GGLRS} for monotonicity testing
  in the standard model, it is easy to obtain a one-sided\footnote{Recall that a one-sided tester for a class of functions is one which must accept (with probability $1$) any function in the class. This is in contrast to making two-sided error, where an algorithm may reject a function in the class with small probability.}, non-adaptive\footnote{A non-adaptive algorithm is one in which the choice of its $i$-th query point does not depend on the responses received to queries $1,\ldots, i-1$.}
 algorithm 
  with $O(n/\eps)$ queries in the relative-error model (see \Cref{sec:warmup} as a warmup).
Our main algorithm as stated below,~on the other hand, makes $O(\log (N) /\eps)$ queries, 
 where $N:=|f^{-1}(1)|$ denotes the \emph{sparsity}~of~$f$ and~is \emph{not} given to the algorithm.
Note that one always has $\log (N)/\eps \le n/\eps$ as $N\le 2^n$ trivially but the
  former can be 
  asymptotically lower when $f$ is sparse  (e.g., $\smash{n^{1/3}/\eps}$ when $\smash{N=2^{n^{1/3}}}$). 
  
\begin{restatable}[Testing algorithm]{theorem}{maintheoremone}
\label{main:upperbound}
There is a one-sided adaptive algorithm which is an $\eps$-relative-error tester for monotonicity (i.e., it always returns ``monotone'' when the input function $f$ is monotone
and returns ``not monotone'' with probability at least $2/3$ when $f$
  has relative distance at least $\eps$ from monotone); with probability at least $1-\delta$, it makes no more than $O(\log(1/\delta)/\eps + \log(N)/\varepsilon)$ calls to $\SAMP(f)$ and $\MQ(f)$.
\end{restatable}
  
We remark that while the algorithm described in \Cref{main:upperbound} is adaptive, 
  it can be made non-adaptive and still makes $O(\log (N)/\eps)$ queries (in the worst case)
  when an estimate of 
  $N$ is given as part of the input to the algorithm.
Our first main lower~bound shows that this is indeed tight:
  $\tilde{\Omega}(\log N)$ queries are needed for non-adaptive algorithms even when
  an estimate of $N$ is given. 
  
\begin{theorem}[Non-adaptive lower bound] \label{main:firstlowerbound}

For any constant $\alpha_0<1$, there exists a constant $\eps_0>0$ such that any two-sided, non-adaptive algorithm for testing whether a function $f$ with $|f^{-1}(1)|=\Theta(N)$ for some given parameter $\smash{N\le 2^{\alpha_0 n}}$ is monotone or has
  relative~distance~at least $\eps_0$ from monotone must make $\tilde{\Omega}(\log N)$ queries.
  \end{theorem}

Finally, we show that $\tilde{\Omega}((\log N)^{2/3})$ queries are needed for 
  adaptive algorithms:

\begin{theorem}[Adaptive lower bound] \label{main:adaptivelowerbound}
For any constant $\alpha_0<1$, there exists a constant $\eps_0>0$ such that any two-sided, adaptive algorithm for testing whether a function $f$ with $|f^{-1}(1)|=\Theta(N)$ for some given parameter $\smash{N\le 2^{\alpha_0 n}}$ is monotone or has
  relative~distance~at least $\eps_0$ from monotone must make $\tilde{\Omega}((\log N)^{2/3})$ queries.
\end{theorem}
  
\Cref{main:firstlowerbound} and \Cref{main:adaptivelowerbound}
  also imply that, as functions of $n$,
  $\tilde{\Omega}(n)$ and $\tilde{\Omega}(n^{2/3})$ queries 
  are needed for non-adaptive and adaptive relative-error monotonicity testing algorithms,
  respectively.  
  
Summarizing our results, we see that there are both a correspondence and a significant point of difference between the state-of-the-art for standard-model monotonicity testing upper and lower bounds \cite{CWX17stoc,KMS18} and the upper and lower bounds for relative-error monotonicity testing that we establish. (We fix the error parameter $\eps$ to be a constant, for simplicity, in the discussion below.)
In both cases there is a 3/2-factor-gap between the exponent of the best known adaptive lower bound and the best known algorithm, but the source of this gap is intriguingly different between the two cases. In the standard model, the best known algorithm is the sophisticated ``path tester'' of Khot et al.~\cite{KMS18} which makes $\smash{\tilde{O}(\sqrt{n})}$ queries, while the strongest lower bound (for general testers, i.e.~adaptive testers with two-sided error) is $\tilde{\Omega}(n^{1/3})$, due to Chen et al.~\cite{CWX17stoc}.  In contrast, in the relative-error setting the best algorithm we know of is a variant of the simple ``edge tester'' (given in \Cref{sec:upper-bound}), which achieves an $O(\log(N))$ complexity.  But in the relative-error setting, for a wide range of values of $N=|f^{-1}(1)|$, it is possible to prove a \emph{stronger} $\tilde{\Omega}((\log N)^{2/3})$ lower bound than in the standard model.  So the gap in the exponent is again a factor of 3/2; but in the relative-error setting the state-of-the-art algorithm is a simple one (in contrast with the sophisticated algorithm and analysis, based on isoperimetry, from \cite{KMS18}), and a stronger lower bound is achievable in the relative-error model than in the standard model.

\subsection{Technical Overview}  

We now discuss techniques underlying our upper and lower bounds for relative-error monotonicity testing, with the goal of giving some intuition for how the bounds are proved.

\paragraph{Upper bounds.} We begin with the relative-error monotonicity testing algorithms. Our algorithm is based on the ``edge tester'' for the standard monotonicity testing problem; the high-level idea is to look for an explicit \emph{edge} of the Boolean hypercube that witnesses a violation of monotonicity.

We begin by recalling the standard edge tester of \cite{GGLRS}. It simply repeats the following test $O(n/\eps)$ times: sample an $\bx$ uniformly at random from $\{0,1\}^n$, sample a $\by$ ``immediately above'' $\bx$ in the Boolean hypercube\footnote{Formally, $\by$ is a point with $\bx \prec \by$ and $\|\by\|_1 = \|\bx\|_1 + 1$, where $\|x\|_1$ denotes the Hamming weight of $x\in \{0,1\}^n$.}, and reject if $f(\bx) > f(\by)$. Note that such a rejection only happens if $(\bx,\by)$ violates monotonicity, so the algorithm will never reject a monotone function $f$. To show that the algorithm will reject any $f$ that is far from monotone with high probability, \cite{GGLRS} showed that if at least $\eps 2^n$ points need to be changed to make $f$ monotone, then $f$ must have at least $\eps 2^{n-1}$ such violating edges. Since there are $n 2^{n-1}$ edges in total in the Boolean hypercube, the probability that the edge tester hits such a violating edge is at least $\eps/n$.

We first describe a relative-error testing algorithm that is assumed to know the sparsity $N$ (i.e.~the sparsity $N$ is given to it as an input). This algorithm first samples an $\bx \in f^{-1}(1)$ using the sampling oracle $\SAMP(f)$, then samples a point $\by$ that is immediately above $\bx$, and rejects if $f(\by) = 0$. One can see that the only difference between this algorithm and the classical edge tester is that instead of sampling $\bx$ uniformly from all of $\{0,1\}^n$, we sample $\bx$ uniformly from $f^{-1}(1)$. Our analysis shows that the algorithm will find a violation with probability at least $\Omega(\eps/(\log N))$ when $f$ is  $\eps$-relative-error-far from monotone.

The key observation for the analysis is the following: if a monotone function $f$ has sparsity $N$, then any point $x \in f^{-1}(1)$ must have $\|x\|_1 \ge n - \log N$, since every point above $x$ must also be in $f^{-1}(1)$. For each $x$ having $\|x\|_1 \ge n - \log N$, there will be at most $\log N$ many possible~$y$'s ``immediately above'' $x$, and hence our relative-error testing algorithm is sampling from at most $N \log N$ possible $(x,y)$ pairs. Applying the results of \cite{GGLRS}, we know that when $f$ is $\eps$-relative-error-far from monotone, $f$ has at least $\eps N/2$ many violating $(x,y)$ edges. Therefore, the probability that our algorithm will find a violation is at least $\eps/(2\log N)$.

Now we consider the more challenging situation that the algorithm does not know the sparsity $N$. 
Our approach is to first obtain (an estimate of) the minimum Hamming weight of all points in $f^{-1}(1)$ by drawing a few samples from the sampling oracle, and then run the algorithm sketched above. In more detail, suppose that $x$ is the sampled point with minimum Hamming weight, which is $n-k$; our algorithm samples a few uniform random points from all points that are above $x$. If any of them have $f=0$ then the algorithm rejects, and otherwise it uses $k$ as a proxy for $\log N$ in the algorithm sketched above. To see why this works, it suffices to consider the case that $f$~is  $\eps$-relative-error-far from monotone.  Observe first that if any sampled $x \in f^{-1}(1)$ has weight $\|x\|_1 < n - 2 \log N$, then most points above $x$ cannot be in $f^{-1}(1)$ (since there are $N^2$ points above $x$ and only $N$ satisfying assignments in total), so sampling a few uniform random points above $x$ will reveal a violation of monotonicity.  So we may assume that the sampled point $x$ with minimum Hamming weight has Hamming weight at least $n - 2 \log N$.  Now our earlier analysis shows that the earlier algorithm will indeed reject any $f$ that has relative error at least $\eps$ for monotonicity.

\paragraph{Lower bounds.}
We consider the case when the relative distance parameter $\eps$ is a constant.
The upper bounds show that even when $N$ is much smaller than $2^n$, 
   there are algorithms for relative error testing that make 
   only $O(\log N)$ queries.
We describe some of the key new ideas behind our 
  constructions for the $\tilde{\Omega}(\log N)$ lower bound for 
  non-adaptive algorithms (\Cref{main:firstlowerbound}, \Cref{sec:two-sided-non-adaptive-lb})
  and the $\tilde{\Omega}((\log N)^{2/3})$ lower bound for adaptive 
  algorithms (\Cref{main:adaptivelowerbound}, \Cref{sec:two-sided-adaptive-lb}).

Intuitively, to make monotonicity testing as hard as possible, we would like to use functions that have as little poset structure of the cube as possible. 
(For example, all monotonicity lower bounds in the standard model use
  functions that are only nontrivial in the \emph{middle layers}, i.e., points
  with Hamming weight $(n/2)\pm O(\sqrt{n})$, with all points above
  the middle layers set to $1$ and all points below set to $0$ by default.
This way only the poset structure in the middle $O(\sqrt{n})$ layers is relevant.)
Towards this end, our lower bound constructions use functions which are only nontrivial on two adjacent layers of the Boolean hypercube, we call them \emph{two-layer} functions
  (see \Cref{sec:two-layer}).
More precisely, a function $f$ is a two-layer function if it is only nontrivial
  on points in layers $3n/4$ and $3n/4+1$: every point with weight $<3n/4$ is
  set to $0$ and every point with weight $>3n/4+1$ is set to $1$.
(We use the constant $3/4$ just to make the presentation more concrete;
  it can be replaced by any constant strictly between $1/2$ and $1$).

In addition to the much simplified poset structure (which we discuss in
  the paragraph about our construction below), two advantages of using two-layer functions follow from simple calculations:
(1) For $f$ to have $\eps$-relative error from monotonicity for some constant $\eps$,
  it suffices to show that a constant fraction of points in the two layers
  $3n/4$ and $3n/4+1$ must be changed to make the function monotone;
(2) The use of two-layer functions helps us reduce the analysis of general relative-error testing algorithms (which can use both $\MQ$\ and $\SAMP$\ oracles) to algorithms that use $\MQ$\ only (see \Cref{claim:two-layer1}).

 So the main technical challenge is to establish strong lower bounds for two-layer functions.  Recall that for monotonicity testing in the standard model, state of the art lower bounds are obtained using the Talagrand DNF \cite{BB15,CWX17stoc} (a random $2^{\sqrt{n}}$-term DNF in which each term has $\sqrt{n}$ variables sampled uniformly at random) and a depth-3 extension of the Talagrand DNF \cite{CWX17stoc}.
 The $\sqrt{n}$ (in both the size of each term and the exponent of the number 
  of terms) comes from the fact we mentioned early about the middle $O(\sqrt{n})$ 
  layers: if $\sqrt{n}$ were replaced by a larger polynomial $n^{0.5+c}$, then
  most layers in the middle would become trivial (points in layers above $(n/2)+O(n^{0.5-c})$ would be most likely set to $1$, and points in layers below
  $(n/2)-O(n^{0.5-c})$ would be most likely set to $0$).

In contrast, since we only focus on two-layer functions, the obstacle 
  described above is no longer there, and it turns out that 
  we can use a novel variant of the Talagrand DNF with $2^{\Theta(n)}$
  many terms each of size $\Theta(n)$.
This is the construction that leads to our non-adaptive $\tilde{\Omega}(\log N)$
  lower bound for relative-error monotonicity testing.
For the $\tilde{\Omega}((\log N)^{2/3})$ adaptive lower bound, 
  our construction uses 
  a corresponding variant of the depth-3 extension of Talagrand DNF.
Using two-layer functions also necessitates various other technical modifications of the \cite{BB15,CWX17stoc} constructions, which are detailed in \Cref{sec:two-sided-non-adaptive-lb} and \Cref{sec:two-sided-adaptive-lb}.

\subsection{Related work}

Several earlier works have considered property testing models which are similar to our model of relative-error testing of Boolean functions $f: \zo^n \to \zo$. One such work is
an early paper of Rademacher and Vempala \cite{RV04}, who considered essentially a continuous-domain analogue of our framework. \cite{RV04} studied the problem of testing whether an unknown set $S \subset \R^n$ is convex versus $\eps$-far from convex, where a finite-volume set $S \subseteq \R^n$ is said to be $\eps$-far from convex if $\Leb(S \ \triangle \ C) \geq \eps \cdot \Leb(S)$ for every convex set $C$, where $\Leb(\cdot)$ denotes the Lebesgue volume.  The query model considered in that work, like our query model, is that the algorithm is given access both to a black-box oracle $\MQ(S)$ for (the indicator function of) the unknown set $S$, as well as access to a $\SAMP(S)$ oracle which when queried outputs a uniform random element of $S$. 

 \cite{RV04} gave an $(n/\eps)^{O(n)}$-complexity algorithm for testing convexity in their model, and also gave an exponential lower bound for a specific convexity tester known as the ``line segment tester''; the \cite{RV04} lower bound was strengthened and extended to an exponential lower bound for a ``convex hull tester'' in recent work of Blais and Bommireddi \cite{BB20}.  Neither of the works \cite{RV04,BB20} studied general algorithms for relative-error property testing or considered relative-error property testing of functions over the Boolean domain $\zo^n$, which is the focus of the present work.
 
A similar model, in which the relative-distance measure is employed and the testing algorithm has both $MQ(f)$ and $\SAMP(f)$ oracles, was considered by Ron and Tsur \cite{RonTsur14}, who considered the problem of testing \emph{sparse images} for properties such as sparsity, convexity, monotonicity, and being a line. In their setting, the domain of interest was the two-dimensional discrete grid $[n] \times [n]$, rather than the high-dimensional Boolean hypercube $\zo^n$ that we consider, and hence the results and techniques are very different between their setting and ours.

We remark that several early works \cite{M00,KMS03} studied self-testing of \emph{real-valued} functions with a relative-error criterion, in the sense that error terms are allowed to be proportional to the particular function value being computed. This line of work, with its focus on real-valued functions, has a rather different flavor from our relative-error model of testing (sparse) Boolean functions.

\subsection{Future work}

Apart from the general observations about relative-error testing given in \Cref{sec:general-results}, this work focuses on monotonicity testing. We hope  that the relative-error model may open up a new perspective on Boolean function property testing more broadly, though, by giving a clean theoretical framework for studying testing of \emph{sparse} Boolean functions for all sorts of properties. Beyond monotonicity, there are many interesting and natural questions about relative-error testing of other specific well-studied Boolean function properties; a few of the questions which seem most interesting to us are listed below.

\begin{flushleft}\begin{itemize}
\item
{\bf Juntas:} As we will discuss after \Cref{obs:relative-from-standard}, there is a relative-error testing algorithm for the class of $k$-juntas that makes $O(k2^k/\eps)$ black-box queries and does not use the sample oracle.  This dependence on $k$ is exponentially worse than the state-of-the-art $O(k/\eps + k \log k)$-query $k$-junta testing algorithm \cite{Blaisstoc09} for the standard model, so it is natural to ask:  can the class of $k$-juntas be relative-error tested using $\poly(k,1/\eps)$ queries?  We note that in the distribution-free model juntas are testable with $\poly(k,1/\eps)$ queries \cite{CLSSX18,Bshouty19}, but it is not clear how a distribution-free testing algorithm can help with relative-error testing.

\item
{\bf LTFs:} As we will mention in \Cref{sec:general-results}, the class of LTFs can be $\eps$-relative-error tested using $\poly(n,1/\eps)$ samples and queries (by a reduction to the $\poly(n,1/\eps)$-sample relative-error learning algorithm of \cite{DDS15}). However, in the standard property testing model it is known that the class of all linear threshold functions (LTFs) over $\{0,1\}^n$ can be $\eps$-tested using only $\poly(1/\eps)$ queries, independent of the ambient dimension $n$ \cite{MORS:10}.    Are LTFs relative-error testable using $\poly(1/\eps)$ queries and samples, or can an $\omega_n(1)$ lower bound be shown?

\item  {\bf DNFs:}  A similar question can be asked for DNFs.  As we will mention in \Cref{sec:general-results}, the class of $s$-term DNFs can be $\eps$-relative-error tested using $\poly(n^{\log(s/\eps)})$ samples and queries (by a reduction to the $\poly(n^{\log(s/\eps)})$-sample relative-error learning algorithm of \cite{DDS15}). However, in the standard property testing model it is known that the class of all $s$-term DNFs over $\{0,1\}^n$ can be $\eps$-tested using only $\poly(s/\eps)$ queries, independent of $n$ \cite{DLM+:07,CGM11,Bshouty20}.
    Are $s$-term DNFs relative-error testable using $\poly(s/\eps)$ queries and samples, or can an $\omega_n(1)$ lower bound (or even an $n^{\omega_s(1)}$ lower bound) be shown?
\end{itemize}\end{flushleft}

%% file: sections/prelims.tex

\section{Preliminaries}
\label{sec:prelims}

We start by defining the distance metric that we will use:

\begin{definition}
Given two functions $f,g: \zo^n \to \zo$, the \emph{relative distance from $f$ to $g$} is defined as
\[
\reldist(f,g) := {\frac {|f^{-1}(1) \ \triangle \ g^{-1}(1)|}{|f^{-1}(1)|}}.
\]
For a class  ${\cal C}$ of functions from $\zo^n$ to $\zo$, the \emph{relative distance from $f$ to ${\cal C}$} is given by
\[
\reldist(f,{\cal C}) := \min_{g \in {\cal C}} \reldist(f,g).
\]
\end{definition}

Note that the relative distance between two Boolean functions is not symmetric (i.e. $\reldist(f,g)$ is not necessarily equal to $\reldist(g,f)$), this does not pose a problem for us: we will chiefly be interested in cases where $\reldist(f,g)$ is small, and it is easily verified that if $\reldist(f,g) = \eps$ is at most (say) $0.99$, then $\reldist(g,f) = O(\eps)$. We further mention that in the regime of our interest,  the definition of relative distance is fairly robust:  as long as $\reldist(f,g)$ is at most a small constant, then replacing $|f^{-1}(1)|$ in the denominator by any of $|g^{-1}(1)|$, $|f^{1}(1)| + |g^{-1}(1)|$, or $\max\{|f^{-1}(1)|,|g^{-1}(1)|\}$, only changes the value of $\reldist(f,g)$ by a constant factor.

Let $f: \zo^n \to \zo$ be the unknown Boolean function that is being tested. Recall that a \emph{sampling oracle for $f$} is an oracle which takes no inputs and, each time it is invoked, independently returns a uniform random element of $f^{-1}(1)$.\footnote{If $f$ is the  constant-0 function then the sampling oracle is assumed to return a special $\bot$ symbol; note that if this happens  it completely identifies the function $f$ as the constant-0 function. Hence in our subsequent discussion we implicitly assume that the function $f$ being tested is not the constant-0 function.}  A \emph{black-box oracle for $f$} takes as input an $n$-bit string $x$ and returns $f(x)$; we sometimes refer to a black-box oracle as simple an ``oracle for $f$.''

\begin{definition} [Relative-error property testing]
\label{def:rel-error-testing}
Let ${\cal C}$ be a class of functions from $\zo^n$ to $\zo$.  A \emph{$q$-query $\eps$-relative-error property testing algorithm for ${\cal C}$} is a randomized algorithm $A$ which is given access to a sampling oracle for $f$ and a black-box oracle for $f$, where $f$ may be any (unknown) function from $\zo^n$ to $\zo$. $A$ makes at most $q$ calls to the sampling oracle and at most $q$ calls to the black-box oracle, and has the following performance guarantee:
\begin{itemize}

\item If $f \in {\cal C}$ then with probability at least $2/3$ algorithm $A$ outputs ``accept'';

\item If $\reldist(f,{\cal C}) > \eps$ then with probability at least $2/3$ algorithm $A$ outputs ``reject.''

\end{itemize}
A \emph{one-sided} relative error property testing algorithm is one which outputs ``reject'' with probability 0 unless $f \notin {\cal C}$ (i.e.~if $f \in {\cal C}$ then it outputs ``accept'' with probability 1).
\end{definition}

\begin{remark} \label{remark:samples-then-queries}
We note that without loss of generality a relative-error testing algorithm can be assumed to make all $q$ of its calls to the sampling oracle before making any calls to the black-box oracle.  When we say that a relative-error testing algorithm is ``non-adaptive'', this means that after receiving the results of all of its calls to the sampling oracle, it makes one parallel round of queries to the black-box oracle (so these queries can depend on the result of the calls to the sampling oracle, but the choice of the $i$-th query point for the black-box oracle does not depend on the responses received to queries $1,\ldots, i-1$).
\end{remark}

\begin{remark} \label{remark:mono-needs-both}
Our main focus of interest in this work will be testing the class ${\cal C}$ of all monotone Boolean functions, which we denote $\cmon$.  It is easy to verify that any $2^{o(n)}$-query relative-error testing algorithm for this class must use both the sampling oracle and the black-box oracle (we give a simple argument establishing this in \Cref{ap:mono-needs-both}).  Thus our framework, which allows both the sampling oracle and the black-box oracle, is a ``minimal adequate model'' for relative-error monotonicity testing.
\end{remark}

%% file: sections/general-results.tex

\section{Some general results on relative-error testing} \label{sec:general-results} \label{sec:general-results}

\noindent {\bf Relative-error testing implies standard-model testing.}
We begin with the following straightforward fact, which shows that any property that is efficiently relative-error testable is also efficiently testable in the standard model:

\begin{fact} \label{obs:relative-yields-standard}
Let ${\cal C}$ be any class of functions from $\zo^n \to \zo$.  Suppose that there is a relative-error $\eps$-testing algorithm $T$ for ${\cal C}$ that makes $q(\eps,n)$ many black-box queries and calls to $\SAMP$.  Then there is a standard-model $\eps$-testing algorithm $T'$ for ${\cal C}$ which makes at most $O(1/\eps^2 + {\frac 1 \eps} q(\eps,n))$ calls to the black-box oracle.
\end{fact}

\begin{proof}
Let $\alpha \in [0,1]$ be the value of $\ham(0,{\cal C})$, i.e.~the Hamming distance between the constant-0 function and the nearest function in ${\cal C}$; the algorithm and its analysis will make use of this quantity.  (Note that no queries to the unknown function $f$ are required for the algorithm to determine the value of $\alpha$.)

The algorithm $T'$ works as follows:

\begin{flushleft}\begin{enumerate}
\item It first makes $m_1 = O(1/\eps^2)$ calls to the black-box oracle for $f$ on uniform random inputs; let $\boldm'_1 \leq m_1$ be the number of those queries that have $f(x)=1.$ 
If $\boldm'_1/m_1 \leq \eps/3$, then $T'$ outputs ``accept'' if $\alpha \leq \eps/2$ and outputs ``reject'' otherwise.  If $\boldm'_1/m_1>\eps/3$, then

\item $T'$ draws $m_2 = {\frac C \eps} q(\eps,n)$ independent uniform random points from $\zo^n$ (for a suitable absolute constant $C$) and queries $f$ on each of them. If fewer than $q(\eps,n)$ of the $m_2$ points are satisfying assignments of $f$ then $T'$ halts and fails. Otherwise, 

\item $T'$ uses the first $q(\eps,n)$ of the satisfying assignments as the $\SAMP(f)$ responses that algorithm $T(\eps,n)$ requires. (Recall that by \Cref{remark:samples-then-queries}, we may suppose that $T$ makes all its calls to $\SAMP(f)$ before making any calls to $\MQ(f)$.) $T'$ continues to simulate $T(\eps,n)$ (making at most $q(\eps,n)$ calls to $\MQ(f)$, the same way $T$ does) and returns what $T'$ outputs.
\end{enumerate}\end{flushleft}

The query complexity of $T'$ is clearly as claimed. To establish correctness, we begin by observing that by a simple Chernoff bound, with probability $99/100$ the value of $\boldm_1'/m_1$ is within $\pm \eps/100$ of the true value of $|f^{-1}(1)|/2^n$ (we will use this repeatedly in the following arguments).

\paragraph{ Case 1:  $f \in {\cal C}$.}  Suppose first that $|f^{-1}(1)|/2^n \leq \eps/4.$ In this case, with probability at least 99/100 we have $\boldm'_1/m_1 \leq \eps/4+\eps/100=26\eps/100$. Since $\alpha \leq \ham(f,0) =  |f^{-1}(1)|/2^n \le \eps/4$, in this case $T'$ correctly outputs ``accept'' in Step~1.

So suppose next that $f \in {\cal C}$ and $\eps/4 < |f^{-1}(1)|/2^n  \leq \eps/2.$  In this case we have $\alpha \leq \ham(f,0) \leq \eps/2$ so $T'$ does not output ``reject'' in Step~1.  Even if $T'$ does not output ``accept'' in Step~1, since $|f^{-1}(1)|/2^n \geq \eps/4$, for a suitable choice of the constant $C$ we have that with probability at least $99/100$ the algorithm $T'$ does not fail in Step~2. So the relative-error $\eps$-testing algorithm $T(\eps,n)$ is executed in Step~3 and we invoke its guarantee, which is that it outputs ``accept'' with probability at least $2/3$ (since $f \in {\cal C}$). So in this case as well $T'$ correctly outputs ``accept'' with high probability.

The remaining subcase of Case~1 is that $f \in {\cal C}$ and $|f^{-1}(1)| / 2^n > \eps/2.$
In this case the probability that $T'$ outputs ``reject'' in Step~1 is at most $1/100$ (since this only happens if $\boldm'_1/m_1 \leq \eps/3$), so we may suppose that $T'$ reaches Step~2. In this case since $|f^{-1}(1)|/2^n \geq \eps/2$, for a suitable choice of the constant $C$ we have that with probability at least $99/100$ the algorithm $T'$ does not fail in Step~2. So as in the previous paragraph, the relative-error $\eps$-testing algorithm $T(\eps,n)$ is executed in Step~3, and in this case as well $T'$ correctly outputs ``accept'' with high probability (which can be amplified to any constant probability using standard techniques).

\paragraph{Case 2: $\ham(f,{\cal C}) > \eps$.}  
Suppose first that $|f^{-1}(1)|/2^n \leq E=\eps/4.$ Then in Step~1, with probability $99/100$ we have $\boldm'_1/m_1 \leq \eps/4 + \eps/100=26\eps/100$.  It cannot be the case that $\alpha \leq \eps/2$ (because if $\alpha \leq \eps/2$ then we would have $\ham(f,{\cal C})\leq  \ham(f,0) + \ham(0,{\cal C}) \leq 3\eps/4$, but in Case~2 we have $\ham(f,{\cal C}) > \eps$), so algorithm $T'_1$ correctly outputs ``reject'' with high probability in Step~1.

Suppose next that $\ham(f,{\cal C}) > \eps$ and $\eps/4< |f^{-1}(1)|/2^n \leq \eps/2$. If $\boldm'_1/m_1 \leq \eps/3$ then similar to the previous paragraph we cannot have $\alpha \leq \eps/2$ (because if $\alpha \leq \eps/2$ then we would have $\ham(f,{\cal C}) \leq \ham(f,0) + \ham(0,{\cal C}) \leq \eps$, but in Case~2 we have $\ham(f,{\cal C}) > \eps$), so $T_1$ correctly outputs ``reject'' in Step~1.  If $\boldm'_1/m_1 > \eps/3$ then since $|f^{-1}(1)|/2^n \geq \eps/4$, for a suitable choice of the constant $C$ we have that with probability at least $99/100$ the algorithm $T'$ does not fail in Step~2. So the relative-error $\eps$-testing algorithm $T(\eps,n)$ is executed in Step~3 and we invoke its guarantee, which is that it correctly outputs ``reject'' with probability at least $2/3$ (since $\reldist(f,{\cal C}) = \ham(f,{\cal C}) \cdot {\frac {2^n}{|f^{-1}(1)|}} \geq \ham(f,{\cal C}) > \eps$).  So in this case as well $T'$ correctly outputs ``reject'' with high probability.

The final case when $\ham(f,{\cal C})>\eps$ is that $|f^{-1}(1)| / 2^n > \eps/2.$
In this case the probability that $T'$ outputs ``reject'' in Step~1 is at most $1/100$ (since this only happens if $\boldm'_1/m_1 \leq \eps/3$), so we may suppose that $T'$ reaches Step~2. In this case the arguments of the previous paragraph again give us that $T'$ correctly outputs ``reject'' with high probability.
\end{proof}

We remark that a simple example, given in \Cref{ap:standard-yes-relative-no}, shows that there are properties which are trivially testable in the standard model but very hard to test in the relative-error model. Together with \Cref{obs:relative-yields-standard}, this shows that relative-error testing is a more demanding model than the standard property testing model.

\medskip
\noindent
{\bf Standard-model testing implies relative-error testing, for not-too-sparse properties.} The next result, together with \Cref{obs:relative-yields-standard}, implies that standard-model testing and relative-error testing are essentially equivalent for classes of functions that are ``not too sparse:''

\begin{fact} \label{obs:relative-from-standard}
Let $p=p(n)>0$ and let ${\cal C}$ be any class of functions from $\zo^n \to \zo$ such that every $f \in {\cal C}$ has $|f^{-1}(1)|/2^n \geq p$.
Suppose that there is a standard-model $\eps$-testing algorithm $T$ for ${\cal C}$ that makes $q(\eps,n)$ many black-box queries.  Then there is a relative-error $\eps$-testing algorithm $T'$ for ${\cal C}$ which, when run on any function $f: \zo^n \to \zo$, makes no calls to the sampling oracle and at most $O(1/p + q(p\eps/2,n))$ calls to the black-box oracle.
\end{fact}

\begin{proof}
Algorithm $T'$ first makes $m=O(1/p)$ calls to the black-box oracle for $f$ on uniform random inputs, and rejects if fewer than $3pm/4$ of the queried points are satisfying assignments of $f$.  If $T'$ does not reject in this first phase, it then performs $O(1)$ many runs of the standard-model testing algorithm $T$, each time with closeness parameter $p\eps/2$, and outputs the majority of those runs.

It is clear that the query complexity is as claimed. 
To establish correctness, first consider the case that $f \in {\cal C}.$ Since $|f^{-1}(1)|/2^n \geq p$, the probability that the algorithm rejects in the first phase is at most (say) 1/6 by a standard multiplicative Chernoff bound.  Since $f \in {\cal C}$, the probability that any individual run of $T$ outputs ``reject'' is at most $1/3$, so the probability that the majority of the $O(1)$ runs output ``reject'' is at most (say) $1/6$.  Hence the probability that $f$ is rejected is at most $1/3$ as desired.  

Next, suppose that $\reldist(f,{\cal C}) > \eps$.  One possibility is that $|f^{-1}(1)|/2^n \leq p/2$; if this is the case, then the probability that $T$ rejects in the first phase is at least 2/3 as required. So consider the other possibility, which is that $|f^{-1}(1)|/2^n > p/2.$ We have
\[
\eps < \reldist(f,{\cal C})  = \min_{g \in {\cal C}} {\frac {|f^{-1}(1) \mathop{\triangle} g^{-1}(1)|}{|f^{-1}(1)|}} <
{\frac 2 p} \cdot \min_{g \in {\cal C}} {\frac {|f^{-1}(1) \mathop{\triangle} g^{-1}(1)|}{2^n}}, 
\]
which rearranges to give $\ham(f,{\cal C})\geq p\eps/2$. So even if $T'$ makes it to the second stage, the probability that the majority of $O(1)$ calls to $T$ with closeness parameter $p\eps/2$ output ``accept'' is at most $1/3$.
\end{proof}

\Cref{obs:relative-from-standard} sheds light on relative-error testing of some classes that have been intensively studied in the standard model.  For instance, since every non-constant parity function $f$ has $|f^{-1}(1)|/2^n=1/2$, the standard-model $O(1/\eps)$-query testing algorithm for linear (parity) functions \cite{BLR93} yields a relative-error tester with the same $O(1/\eps)$-query complexity.  As another application, since every non-constant $k$-junta has $|f^{-1}(1)|/2^n \geq 1/2^k$, the $O(k/\eps + k \log k)$-query algorithm of Blais \cite{Blaisstoc09} for testing juntas yields a relative-error tester for $k$-juntas that makes $O(k2^k/\eps)$ black-box queries.  On the other hand,  \Cref{obs:relative-from-standard} does not give anything for some other natural classes such as LTFs, $s$-term DNF formulas, and monotone functions.

\medskip
\noindent
{\bf Relative-error learning implies relative-error testing.} Several recent works \cite{DDS15,CDS20soda} have studied the problem of learning the distribution of satisfying assignments of an unknown Boolean function. In this framework, a learning algorithm for a concept class ${\cal C}$ of Boolean functions over $\zo^n$ has access to uniform random satisfying assignments (i.e.~to a sample oracle for the unknown target function $f \in {\cal C}$), and the goal of the learner is to output an \emph{$\eps$-sampler for $f^{-1}(1)$}, which is a circuit that, when given independent uniform random bits as its input, outputs a draw from a distribution ${\cal D}$ that has total variation distance at most $\eps$ from the uniform distribution over $f^{-1}(1)$.  Known algorithms in this framework work in two stages:

\begin{flushleft}\begin{enumerate}
\item First, they perform \emph{$\eps/2$-relative-error proper learning} of the unknown target function $f \in {\cal C}$. This means that they use independent uniform samples from $f^{-1}(1)$ to construct a hypothesis function $h \in {\cal C}$ which, with probability at least $9/10$, satisfies $\reldist(f,h) \leq \eps/2.$ 

\item Next, they output an $\eps/2$-sampler for $h$.
\end{enumerate}\end{flushleft}

We now show that relative-error \emph{learning} algorithms for a class of functions yield relative-error \emph{testing} algorithms for the same class, with query complexity comparable to the sample complexity of the learning algorithm:

\begin{fact} \label{obs:test-from-learn}
Let ${\cal C}$ be a class of functions from $\zo^n$ to $\zo$. 
Let $A$ be an algorithm which performs $\eps$-relative-error proper learning of ${\cal C}$ using $s(\eps,n)$ uniform samples from $f^{-1}(1).$
Then there is an $\eps$-relative-error testing algorithm $T$ for ${\cal C}$ which makes at most $s(\eps/4,n) + O(1/\eps^2)$ calls to the sample oracle and at most $O(1/\eps^2)$ calls to the black-box oracle $\MQ(f).$
\end{fact}

\begin{proof}
The testing algorithm $T$ works in the following stages:
\begin{flushleft}\begin{enumerate}
\item Run the relative-error learning algorithm $A$ with error parameter $\eps/4$, using the sample oracle as the source of uniform random examples (this requires $s(\eps/4,n)$ calls to the sample oracle, and no calls to the MQ oracle). If $A$ does not output a hypothesis $h \in {\cal C}$ then output ``reject.''  Otherwise, let $h \in {\cal C}$ denote the hypothesis that $A$ outputs, and

\item Draw $m := O(1/\eps^2)$ samples from $\SAMP(f)$ and evaluate $h$ on each of them.  If $h$ evaluates to 0 on more than $(3/8) \eps m$ of the samples, then output ``reject.'' Otherwise,

\item Draw $m$ independent uniform samples from $h^{-1}(1)$\footnote{Note that while it may be a computationally hard task to generate a uniform random satisfying assignment of $h$, we are only concerned with the query complexity of our testing algorithm and not its running time.} and use $\MQ(f)$ to evaluate $f$ on each of them. If $f$ evaluates to 0 on more than $ (3/8) \eps m$ of them then output ``reject,'' otherwise output ``accept.''
\end{enumerate}\end{flushleft}

The query complexity is clearly as claimed, so we turn to establishing correctness.  Suppose first that the target function $f$ belongs to ${\cal C}$. 
The probability that the $(\eps/4)$-relative-error learning algorithm does not output a function $h \in {\cal C}$ is at most $1/10$, so suppose that $h \in {\cal C}$ is the output of the relative-error learning algorithm, and that $\reldist(f,h) \leq \eps/4.$
Write 
\begin{equation} \label{eq:abc}
a \cdot 2^n := |f^{-1}(1) \cap h^{-1}(1)|, \quad \quad
b \cdot 2^n := |f^{-1}(1) \setminus h^{-1}(1)|, \quad \quad
c \cdot 2^n := |h^{-1}(1) \setminus f^{-1}(1)|.
\end{equation}
Since $\reldist(f,h)={\frac {|f^{-1}(1) \mathop{\triangle} h^{-1}(1)|}{|f^{-1}(1)|}} \leq \eps/4$, we have
${\frac {b+c}{a+b}} \leq \eps/4$. Since the probability that $h$ evaluates to 0 on a random example drawn from $\SAMP(f)$ is ${\frac b {a+b}} \leq {\frac {b+c}{a+b}} \leq \eps/4$, by a standard Chernoff bound the probability that $T$ outputs ``reject'' in Step~2 is at most $1/10$; so suppose that algorithm $T$ proceeds to Step~3.  Since the probability that $f$ evaluates to 0 on a random example drawn from $h^{-1}(1)$ is 
\[\frac c{a + c} = \frac{a + b}{a + c} \cdot \frac c{a + b} \le \frac{a + b}a \cdot \frac\varepsilon4 = \frac1{1 - b/(a + b)} \cdot \frac\varepsilon4 \le \frac1{1 - \varepsilon/4}\cdot \frac\varepsilon4 \le \frac\varepsilon3,\]
the probability that $T$ outputs ``reject'' in Step~3 is at most $1/10$. So the overall probability that $T$ outputs ``reject'' is at most $1/10 + 1/10 + 1/10 < 1/3$, as required since $f \in {\cal C}$.

Next, suppose that $\reldist(f,{\cal C}) > \eps.$ If the relative-error learning algorithm does not output a hypothesis $h \in {\cal C}$, then the tester $T$ outputs ``reject'' (as desired), so suppose that in stage~1 the tester $T$ outputs a hypothesis $h \in {\cal C}$. Let $a,b,c \in [0,1]$ be as in \Cref{eq:abc}. Since ${\frac {b+c}{a+b}} = \reldist(f, h) \ge \reldist(f,{\cal C}) > \eps$, it must be the case that either ${\frac b {a+b}}$ (which is the fraction of examples in $f^{-1}(1)$ that are labeled 0 by $h$) is at least $5\varepsilon/13$, or ${\frac c {a+b}}$ is at least $8\eps/13$. In the first case, a Chernoff bound gives that $T$ outputs ``reject'' in Step~2 with probability at least $2/3$, and in the second case,
\[\frac c{a + c} = \frac1{a/c + 1} \ge \frac1{(a + b)/c + 1} \ge \frac1{13/(8\varepsilon) + 1} \ge \frac{8\varepsilon}{21}.\]
Note that $\frac c{a + c}$ is the fraction of examples in $h^{-1}(1)$ that are labeled 0 by $f$, so a Chernoff bound gives that $T$ outputs ``reject'' in Step~3 with probability at least $2/3.$
\end{proof}

\Cref{obs:test-from-learn} lets us obtain some positive results for relative-error testing of well-studied concept classes, namely LTFs and DNFs, from known positive results for learning those classes:
\begin{itemize}

\item
En route to giving an efficient algorithm for learning the uniform distribution over satisfying assignments of an unknown LTF over $\zo^n$, \cite{DDS15} gives a $\poly(n,1/\eps)$-sample algorithm which is an $\eps$-relative-error proper learner for LTFs over $\zo^n$. Hence by \Cref{obs:test-from-learn}, we get a $\poly(n,1/\eps)$-query algorithm for relative-error testing of LTFs over $\zo^n$. 

\item 
\cite{DDS15} also gives a $\poly(n^{\log(s/\eps)})$-sample $\eps$-relative-error proper learning algorithm for $s$-term DNFs over $\zo^n$ (as part of an algorithm for learning the uniform distribution over satisfying assignments of an unknown $s$-term DNF).  Consequently, \Cref{obs:test-from-learn} gives a $\poly(n^{\log(s/\eps)})$-query algorithm for $\eps$-relative-error testing of the class of $s$-term DNFs.

\end{itemize}

On the other hand, \Cref{obs:test-from-learn} is not useful for relative-error monotonicity testing, since even under the uniform distribution on $\zo^n$(where both positive and negative random examples are available to the learner) no algorithm can learn monotone Boolean functions to accuracy $\eps$ using fewer than $2^{\Omega(\sqrt{n}/\eps)}$ samples \cite{BCOST15}.

%% file: sections/upper-bounds.tex

\section{Algorithms for relative-error monotonicity testing} \label{sec:upper-bound}

In this section we present algorithms for relative-error monotonicity testing. At a high level, our algorithms uses the same idea as the $O(n/\epsilon)$-query edge tester for standard monotonicity testing. However, our algorithms will be much more efficient when the function being tested is \emph{sparse}.
In particular, let $N$ be the sparsity of the given function $f$, i.e.~$N = \lvert f^{-1}(1) \rvert$. Our algorithms make $O(\log(N)/\epsilon)$ calls to $\MQ(f)$ and $\SAMP(f)$, where $\epsilon$ is the (relative-error) distance parameter. 
Note that such an upper bound can still be highly non-trivial vis-a-vis standard monotonicity testing even when $N = 2^{\Theta(n)}$; for example, if $N=2^{0.99n}$, then having relative-error $\eps$ is a much stronger guarantee than having absolute error $\eps$, since $\ham(f,\cmon)=\reldist(f,\cmon)\cdot 2^{-0.01n}$.  

We will present two versions of our algorithm that work in different settings. As a warmup, we first explain how a variant of the simple edge tester gives an $O(n/\eps)$-query algorithm; this algorithm does not need to be provided with the value of $N$. Then, for our first main algorithm, we show how a modification of this simple algorithm works, using $O(\log(N)/\eps)$ queries, if the sparsity $N$ is given to the algorithm as an input parameter. 
Finally, our second main algorithm does not require the sparsity $N$ as an input parameter.  
It works for any Boolean function $f: \{0,1\}^n \to \{0,1\}$, and as long as the function $f$ has sparsity $N$, the algorithm will (usually) terminate within $O(\log (N)/\epsilon)$ steps. The idea of this algorithm is to first \emph{learn} a rough estimate of the sparsity $N$ and then run the first algorithm using the learned estimate of $N$. Because of this, the second algorithm is adaptive, and it terminates after $O(\log(N)/\epsilon)$ samples and queries only with high (say $0.99$) probability. (We will show how to boost the probability to an arbitrarily large $1-\delta $ with only an $O(\log (1/\delta)/\eps)$ additive increase in the sample and query complexity.)

\subsection{Warmup:  A nonadaptive $O(n/\eps)$-query algorithm (that is not given $N$)}\label{sec:warmup}

We first show that the $O(n/\eps)$ edge-tester algorithm in \cite{GGLRS} still works in the relative-error setting with slight modification. The intuition of the edge-tester is simple: the algorithm just tries to find a \emph{violating edge}, i.e.~an edge $\{x, y\}$ such that 
$((x \prec y) \wedge (f(x) > f(y)))$,
and outputs ``not monotone'' if it  finds such an edge. Since any far-from-monotone function must have  ``many'' violating edges, an algorithm that samples some edges and outputs ``not monotone'' if and only if there is a violating edge among the sampled edges should work.

In more detail, we recall the following:

\begin{lemma}[{\cite[Theorem 2]{GGLRS}}]\label{lem:vio-edges-lb}
    If $f: \{0, 1\}^n \to \{0, 1\}$ differs from any monotone function on at least $\Delta$ points  in $\zo^n$, then the number of violating edges $\{x, y\}$ for $f$ is at least $\Delta/2$.
\end{lemma}

\begin{remark} \label{remark:vio-edges}
    The above lemma is tight up to constant factors \cite[Proposition 4]{GGLRS}.
\end{remark}

If $f$ satisfies $\reldist(f,\cmon) \geq \varepsilon$ and $|f^{-1}(1)| = N$, then $f$ differs from any monotone function on at least $\varepsilon N$ points. Thus, it directly follows from \Cref{lem:vio-edges-lb} that there are at least $\varepsilon N/2$ violating edges for $f$. Moreover, all violating edges have one vertex $x$ such that $f(x) = 1$, so in order to find violating edges, we can sample from the set of all edges that have at least one vertex $x$ such that $f(x)$, and this can be done with the sampling oracle. We thus have the following algorithm:

\begin{mdframed}
    \begin{center}
        \textbf{\setword{Algorithm 1}{alg1}}: An edge-testing algorithm
    \end{center}
    \hspace{0.1cm}
    
    \textbf{Input}: $n$, $\eps$, and access to the oracles
   $\MQ(f)$ and $\SAMP(f)$.
    
    \textbf{Output}: Either ``monotone'' or ``not monotone''.
    \begin{flushleft}\begin{enumerate}
        \item Query the sample oracle $\SAMP(f)$ for $a := 4n/\varepsilon$ times.
        \item For every sample $x$ received, uniformly sample $i \in [n]$, and query the oracle $\MQ(f)$ on $x \oplus e_i$ (which is $x$ with its $i$-th bit flipped). If $x_i = 0$ (so that $x \prec x \oplus e_i$) and $f(x \oplus e_i) = 0$ (which means $\{x, x \oplus e_i\}$ is a violating edge), then output ``not monotone'' and halt. 
        \item If the algorithm did not output ``not monotone'' in Step~2, then output ``monotone''.
    \end{enumerate}\end{flushleft}
\end{mdframed}

The number of oracle calls made by the algorithm is $2 \times 4n/\varepsilon = O(n/\varepsilon)$. It is easy to see that the algorithm is non-adaptive, as the queries in Step 2 do not depend on one another and can thus be made in one round. Also, it is a one-sided algorithm, since when $f$ is monotone, there is no violating edge, so the algorithm always outputs ``monotone''. 

Now consider the case where $f$ is $\varepsilon$-relatively-far from monotone. Note that all violating edges are in the form of $\{x, x \oplus e_i\}$ where $f(x) = 1$, so for each $(x, i)$ uniformly sampled from $f^{-1}(1) \times [n]$, the probability that $\{x, x \oplus e_i\}$ is a violating edge is at least $(\varepsilon N/2)/(N \cdot n) = \varepsilon/(2n)$, as there are at least $\varepsilon N/2$ violating edges. Since the algorithm gets $a$ uniform samples $(x, i)$ from $f^{-1}(1) \times [n]$, the probability that none corresponds to a violating edge is at most
\[\left(1 - \frac\varepsilon{2n}\right)^a = \left(1 - \frac\varepsilon{2n}\right)^{4n/\varepsilon} \le \frac13\]
for sufficiently large $n$. So with probability at least $2/3$, the algorithm outputs ``not monotone''.

\subsection{A non-adaptive $O(\log(N)/\eps)$-query algorithm that is given $N$} \label{sec:A}

In this section, we will improve the algorithm in the previous subsection and prove the following:

\begin{theorem} \label{thm:alg-nonadaptive}
There is a one-sided non-adaptive algorithm which, if it is given $N := |f^{-1}(1)|$, is an $\eps$-relative-error tester for monotonicity making $O(\log(N)/\eps)$ calls to $\SAMP(f)$ and $\MQ(f)$.
\end{theorem}

The key to this theorem is the following observation: if $f$ is monotone and $|f^{-1}(1)| = N$, then for any $x \in f^{-1}(1)$, the value of $\|x\|_1$ (i.e.~the Hamming weight of $x$) is at least $n - \log_2N$, because $f(y) = 1$ for any $y \succ x$. Therefore, if we know $N$ and we found a point $x\in \{0,1\}^n$ with $f(x)=1$ and  $\|x\|_1 < n - \log_2N$, then $f$ cannot be monotone.

Moreover, if we assume that every $x$ we get from the sample oracle has Hamming weight close to $n$, then we can also improve the edge-testing algorithm (see~\ref{alg1}). Specifically, in the edge-testing algorithm, we sample violating edges by sample from all edges that have one vertex $x$ in $f^{-1}(1)$. However, for $x \in f^{-1}(1)$, only edges going ``upwards'' from $x$ (connecting $x$ and $x \oplus e_i$ where $x \prec x \oplus e_i$) can be violating edges. Therefore, it suffices to sample only from edges going upwards from $x$, and moreover, the number of these edges is much smaller than $n$ since $x$ has Hamming weight close to $n$.

These observations naturally lead to the following algorithm.

\medskip

\begin{mdframed}
    \begin{center}
        \textbf{\setword{Algorithm 2}{alg2}}: A non-adaptive algorithm that is given $N$
    \end{center}
    \hspace{0.1cm}
    
    \textbf{Input}: $n$, $\eps$, $N$, and access to the oracles $\MQ(f)$ and $\SAMP(f)$.
    
    \textbf{Output}: Either ``monotone'' or ``not monotone''.
    \begin{enumerate}
        \item Query the sample oracle $\SAMP(f)$ for $a := 16\log_2 (N) /\varepsilon$ times.
        \item For every sample $x$ received:
        \begin{flushleft}\begin{itemize}
            \item If $\|x\|_1 < n - 2\log_2 N$, output ``not monotone'' and halt.
            \item Otherwise, uniformly sample $i$ from $\{i \in [n] : x_i = 0\}$, and query the oracle $\MQ(x)$ on $x \oplus e_i$. If $f(x \oplus e_i) = 0$ (which means $\{x, x \oplus e_i\}$ is a violating edge), then output ``not monotone'' and halt. 
        \end{itemize}\end{flushleft}
        \item If the algorithm did not output ``not monotone'' in Step~2, then output ``monotone''.
    \end{enumerate}
\end{mdframed}

\medskip

The number of oracle calls made by the algorithm is $2 \times 16\log_2(N)/\varepsilon = O(\log(N)/\varepsilon)$. Similar to the previous subsection, the algorithm is non-adaptive, as the queries in Step 2 do not depend on one another and can thus be made in one round, and the algorithm is also one-sided, since when $f$ is monotone, there is no violating edge, so the algorithm always outputs ``monotone''.

Now consider the case where $f$ is $\varepsilon$-relatively-far from monotone. Let $g$ be the closest (under Hamming distance) monotone function from $f$, and let $S = f^{-1}(1) \ \triangle \ g^{-1}(1)$, i.e. the set of points where $f$ and $g$ differ. Note that $|S| \ge \varepsilon N$. We have the following two cases:
\begin{flushleft}\begin{itemize}
    \item {\bf Case 1:}  There are at least $\varepsilon N/2$ elements of $S$ that have Hamming weight less than $n - 2\log_2N$. In this case, the probability that \ref{alg2} draws a sample that has Hamming weight less than $n - 2\log_2N$ and outputs ``not monotone'' is at least
    \[1 - \left(1 - \frac{\varepsilon N/2}N\right)^a = 1 - \left(1 - \frac\varepsilon2\right)^{16\log_2(N)/\varepsilon} \ge 2/3.\]

    \item {\bf Case 2:} There are more than $\varepsilon N/2$ elements of $S$ that has Hamming weight at least $n - 2\log_2N$. In this case, if one of the $x$ sampled in Stage~1 of the algorithm has Hamming weight less than $n - 2\log_2N$, then the algorithm outputs ``not monotone''. Below we condition on the event that all $x$ sampled in the algorithm have Hamming weight at least $n - 2\log_2N$.
    
    We first bound the number of violating edges both of whose vertices have Hamming weight at least $n - 2\log_2N$. Define function $f': \{0, 1\}^n \to \{0, 1\}$ such that $f'(t) = f(t)$ if $\|t\|_1 \ge n - 2\log_2N$ and $f'(t) = 0$ otherwise. Then $\ham(f', f) < (\varepsilon N/2)/2^n$, so by the triangle inequality, the number of points on which $f'$ must disagree with any monotone function is greater than $\varepsilon N/2$. Therefore, by \Cref{lem:vio-edges-lb}, there are at least $\varepsilon N/4$ violating edges for $f'$, and thus at least $\varepsilon N/4$ violating edges for $f$ with both vertices having Hamming weight at least $n - 2\log_2N$.

    For each violating edge both of whose vertices have Hamming weight at least $n - 2\log_2N$, the probability that the edge is $\{x, x \oplus e_i\}$ when $x$ is uniformly sampled from $f^{-1}(x)$ and $i$ is uniformly sampled such that $x_i = 0$, conditioning on $\|x\|_1 \ge n - 2\log_2N$, is at least $1/(N \cdot 2\log_2N)$. Since there are at least $\varepsilon N/4$ such violating edges, the probability that \ref{alg2} encounters one of them and thus outputs ``not monotone'' is at least
    \[1 - \left(1 - \frac{\varepsilon N/4}{2N\log_2N}\right)^a = 1 - \left(1 - \frac\varepsilon{8\log_2N}\right)^{16\log_2(N)/\varepsilon} \ge 2/3.\]
\end{itemize}\end{flushleft}

Hence, when $f$ is $\varepsilon$-relatively-far from monotone, \ref{alg2} outputs ``not monotone'' with probability at least $2/3$.

\subsection{An $O(\log(N)/\eps)$-complexity algorithm that does not know $N$}

In this subsection, we consider the case when the algorithm does not have prior knowledge of $N = |f^{-1}(1)|$, and prove the following, which is a restatement of \Cref{main:upperbound}:

\maintheoremone*

For this setting, the intuition is straight-forward: we first somehow get a rough estimate of $N$, and then proceed as in the previous section as if we knew $N$.

Let $\hat k$ be the $(\varepsilon N/2)$-th-least Hamming weight among elements in $f^{-1}(1)$, that is, the $(\varepsilon N/2)$-th smallest item in (the multi-set) $\{\|x\|_1 : x \in f^{-1}(1)\}$. We will estimate $\hat k$, or more precisely, estimate a lower bound for $\hat k$, in the algorithm. Recall that for monotone $f$ and $x \in f^{-1}(1)$, all $y \succ x$ should satisfy $f(y) = 1$, so it is easy to see that $N = |f^{-1}(1)| \ge 2^{n - \hat k}$, and thus intuitively, $\hat k$ gives useful information about the size of $N$ for monotone $f$.

Now we describe the algorithm as in \textbf{Algorithm 3}. Let $\delta \in (0, 1/2)$ be a parameter.

\begin{figure}[t!]
\begin{mdframed}
    \begin{center}
        \textbf{\setword{Algorithm 3}{alg3}}: An algorithm that is not given $N$
    \end{center}
    \hspace{0.1cm}
    
    \textbf{Input}: $n$, $\eps$, and access to the oracles $\MQ(f)$ and $\SAMP(f)$.
    
    \textbf{Output}: Either ``monotone'' or ``not monotone''.
    \begin{flushleft}\begin{enumerate}
       
        \item Query the sample oracle $\SAMP(f)$ for $c := 8\log_2(1/\delta)/\varepsilon$ times.
        Let $z$ be the sample that has the least Hamming weight (if there are multiple such $z$, choose an arbitrary one); set $k := \|z\|_1$. Uniformly sample $y_1, \dots, y_b \in \{y \in \{0, 1\}^n : y \succ z\}$ where $b := 4\log_2(1/\delta)$. If there exists $j \in [b]$ such that $f(y_j) = 0$, output ``not monotone'' and halt.
        
        \item Query the sample oracle $\SAMP(f)$ for $a := 16(n - k)/\varepsilon$ times.
       \item For every sample $x$ received:
        \begin{itemize}
            
            \item Uniformly sample $i$ from $\{i \in [n] : x_i = 0\}$, and query $f(x \oplus e_i)$. If $f(x \oplus e_i) = 0$ (which means $\{x, x \oplus e_i\}$ is a violating edge), output ``not monotone'' and halt. 
        \end{itemize}
        \item If the algorithm has not output ``not monotone'' so far,  output ``monotone''.
    \end{enumerate}\end{flushleft}
\end{mdframed}
\end{figure}

It is easy to see that 
the algorithm always outputs ``monotone'' when $f$ is monotone given that 
  it only returns ``not monotone'' when a violation to monotonicity is found. Below we consider the case when $f$ is $\varepsilon$-relatively-far from monotone
and show that the algorithm returns ``not monotone'' with probability at least $2/3$; we analyze its query complexity at the end of the proof.

We first show that at the end of Step 1, with probability at least $\max\{5/6, 1 - \delta\}$, either the algorithm outputs ``not monotone'' or $k$ satisfies 
\begin{equation}\label{eq:hehe}
n - 2\log_2N \le k \le \hat k.
\end{equation} 
Note that the number of elements in $f^{-1}(1)$ that have Hamming distance at most $\hat k$, by the choice of $\hat k$, is at least $\varepsilon N/2$, so the probability that one of them appears as a sample in the first part of  Step 1 (and thus $k \le \hat k$), is at least 
\[1 - \left(1 - \frac{\varepsilon N/2}N\right)^c \ge 1 - \left(1 - \frac\varepsilon2\right)^{8\log_2(1/\delta)/\varepsilon} \ge \max\left\{\frac{11}{12}, 1 - \frac\delta2\right\}.\] 
On the other hand:
\begin{itemize}
    \item If $\|z\|_1 < n - 2\log_2N$ in Step 1, the probability that  $f(y_j) = 0$ for some $j\in [b]$ is at least 
    \[1 - \left(\frac N{2^{2\log_2N}}\right)^b \ge 1 - \frac1{N^{4\log_2(1/\delta)}} \ge \max\left\{\frac{11}{12}, 1 - \frac\delta2\right\},\]
    in which case the algorithm will output ``not monotone''. 

\end{itemize}
Therefore, by a union bound, with probability at least $\max\{5/6, 1 - \delta\}$, either the algorithm outputs ``not monotone'' in Step 1, or $k$ satisfies \Cref{eq:hehe}.

Assuming \ref{alg3} does not output ``not monotone'' in Step 1 and  $k$ satisfies \Cref{eq:hehe}, 
  we show that with probability at least $5/6$, the algorithm outputs ``not monotone'' in Step 3. 
To this end, we bound the probability of an $x$ from $\SAMP(f)$
  leading Step 3 to return ``not monotone'':

\begin{flushleft}\begin{itemize}
\item 
Define $f': \{0, 1\}^n \to \{0, 1\}$ such that $f'(t) = f(t)$ if $t$ has Hamming weight at least $k$ and $f'(t) = 0$ otherwise. Since $k \le \hat k$, we have $\ham(f', f) \le \varepsilon N/2$, so by triangle inequality, the Hamming distance between $f'$ and any monotone function is at least $\varepsilon N/2$. Therefore, by \Cref{lem:vio-edges-lb}, there are at least $\varepsilon N/4$ violating edges for $f'$, and thus at least $\varepsilon N/4$ violating edges for $f$ with both vertices having Hamming weight at least $k$.
As a result, when $x$ is sampled from $\SAMP(f)$, the probability that Step 3 returns ``not monotone'' is at least
$$
\frac{\eps}{4}\cdot \frac{1}{n-k},
$$
where the $\eps/4$ is the probability that $x$ is in one of these $\eps N/4$ violating edges and 
  $1/(n-k)$ is the probability that the $i$ sampled in Step 3 is exactly the direction of the edge that $x$ lies in.
\end{itemize}\end{flushleft}
 As a result, Step 3 returns ``not monotone'' with probability at least 
    \[1 - \left(1 - \frac{\varepsilon }{4(n - k)}\right)^a = 1 - \left(1 - \frac{\varepsilon}{4(n-k)}\right)^{16(n-k)/\eps} \ge \frac{5}{6},\]
and therefore, when $f$ is $\varepsilon$-relatively-far from monotone, \ref{alg3} outputs ``not monotone'' with probability at least $2/3$.

The number of oracle calls \ref{alg3} makes is at most \[8\log_2(1/\delta)/\eps + 4\log_2(1/\delta) + 2 \times 16(n-k)/\varepsilon = O(\log(1/\delta)/\eps + \log(N)/\varepsilon)\]
when $n - 2\log_2N \le  k$ holds from \Cref{eq:hehe}, which happens with probability at least  $1-\delta$.

%% file: sections/two-sided-nonadaptive-lb.tex

\newcommand{\cyan}[1]{{\color{cyan}{#1}}}
\def\TT{{\bT}}
\def\ff{{\boldf}}
\def\calE{\mathcal{E}}
\def\Dy{\mathcal{D}_{\text{yes}}}
\def\Dn{\mathcal{D}_{\text{no}}}
\def\bGamma{\boldsymbol{\Gamma}}
\def\HH{{\bH}}
\def\hh{{\bh}} 
\def\Ey{\mathcal{E}_{\text{yes}}}
\def\En{\mathcal{E}_{\text{no}}}
\def\frakP{\mathfrak{P}}
\def\ALG{\text{ALG}}
\def\outc{\mathsf{outcome}}

\section{A two-sided non-adaptive lower bound}
\label{sec:two-sided-non-adaptive-lb}

The goal of this section is to prove the following two-sided, non-adaptive lower bound for relative-error monotonicity testing:  

\begin{theorem}~\label{thm:two-sided-non-adaptive-lb-hehe}
Let $\smash{N={n\choose 3n/4}}$. There is a constant $\eps_0>0$ such that any two-sided, non-adaptive algorithm for testing whether a function $f$ with $|f^{-1}(1)|=\Theta(N)$ is monotone or has
  relative~distance~at least $\eps_0$ from monotone functions must make $\tilde{\Omega}(\log N)$ queries.
\end{theorem}

\subsection{A useful class of functions for lower bounds: Two-layer functions} \label{sec:two-layer}

We say $f:\{0,1\}^n\rightarrow \{0,1\}$ is a \emph{two-layer} function if
\begin{equation} \label{eq:two-layer}
    f(x) = \begin{cases} 1 &\textrm{if } \Vert x \Vert_1 > 3n/4 + 1 \\ 
    0 &\textrm{if } \Vert x \Vert_1 < 3n/4 
    \end{cases}
\end{equation}
All functions used in our lower bound proofs in the next two sections
  are two-layer functions.

One reason that two-layer functions are helpful for lower bound arguments is because the $\SAMP(f)$ oracle is not needed for two-layer testing --- it can be simulated, at low overhead, with the $\MQ$ oracle. This means that to prove a lower bound in our relative-error model for two-layer functions, it is sufficient to consider the slightly simpler setting in which the testing algorithm's only access to the unknown two-layer function $f$ is via the $\MQ$ oracle.

\begin{claim} 
\label{claim:two-layer1}
If $f: \zo^n \to \zo$ is a two-layer functions, then for any constant $\tau > 0$, making $q$ calls to the $\SAMP(f)$ oracle can be simulated, with success probability at least $1-\tau$, by making $O(q)$ calls to the $\MQ$ oracle. 
\end{claim}
\begin{proof}
The simulation works by simply making $Cq$ (for a suitable large constant $C$) many $MQ(f)$ queries on independent uniform random points $\bx \in \zo^n$ that satisfy $\|\bx\|_1 \geq 3n/4$, and using the first $q$ points for which $f(\bx)=1.$
It is clear that each such point $\bx$ in the sample (for which $f(\bx)=1$) corresponds to an independent draw from $\SAMP(f)$; correctness follows from a simple probabilistic argument using the fact that for any two-layer function $f$, at least an $\Omega(1)$ fraction of all $n$-bit strings with Hamming weight at least $3n/4$ are satisfying assignments of $f$.
\end{proof}

\begin{remark} \label{remark:two-layer-flexible}
In \Cref{thm:two-sided-non-adaptive-lb-hehe} as well as the rest of the paper,
  we work with two-layer functions as described above, where the two layers in question are $3n/4$ and $3n/4+1$. This choice is made just for concreteness to aid with readability; it is easy to verify that the definition of two-layer functions could be altered to use the two layers $\alpha n$ and $\alpha n + 1$, for any constant $\alpha \in (1/2, 1)$, and that 
  \Cref{thm:two-sided-non-adaptive-lb-hehe} would still go through with $N=\Theta({n\choose \alpha n})$.
\end{remark}

To see that \Cref{thm:two-sided-non-adaptive-lb-hehe} implies \Cref{main:firstlowerbound},
  we first note that for any choice of the parameter $\smash{N\le {n\choose 3n/4}}$, there exists a positive integer $k\le n$
  such that $\smash{N=\Theta({k\choose 3k/4})}$.
The desired $\tilde{\Omega}(\log N)$ lower bound for relative-error testing of functions with sparsity $\Theta(N)$
  can then be obtained 
  from a routine reduction to \Cref{thm:two-sided-non-adaptive-lb-hehe} (with $n$ set to $k$) by
  embedding in a suitable subcube of $\zo^n$ using
  functions 
  $f: \zo^n \to \zo$ of the form $$f(x_1,\dots,x_n) = (x_{k+1} \wedge \cdots \wedge x_n) \wedge f'(x_1,\dots,x_k).$$ 
Moreover, as discussed in \Cref{remark:two-layer-flexible}, $3/4$ can be replaced
  by any constant $\alpha\in (1/2,1)$. Choosing $\alpha$ to be sufficiently close to $1/2$ extends
  the lower bound to any $N\le 2^{\alpha_0 n}$ for any constant $\alpha_0<1$.
  
As it will become clear in \Cref{subsec:dist1}, all functions used in our lower bound proof
  are two-layer functions.
We will prove that any 
deterministic, ``two-stage'' (see \Cref{def:two-stage})
$\MQ$-only algorithm for testing
  two-layer functions with success probability $(99/100)\cdot 2/3$ must make $\tilde{\Omega}(n)$ queries.
As discussed at the end of \Cref{subsec:dist1}, this implies \Cref{thm:two-sided-non-adaptive-lb-hehe}.

\subsection{Distributions $\Dyes$ and $\Dno$}\label{subsec:dist1}
Our proof is an adaptation of the $\tilde{\Omega}(\sqrt{n})$ lower bound from \cite{CWX17stoc}  for two-sided non-adaptive monotonicity testing in the standard model.  

To describe the construction of the yes- and no- distributions $\Dyes$ and $\Dno$, we begin by describing the distribution $\calE$. $\calE$ is uniform over all tuples  $T = (T_i: i\in [L])$, where $L:=(4/3)^n$ and $T_i:[n]\rightarrow [n]$.
Equivalently, to draw a tuple $\TT\sim \calE$,
  for each $i\in [L]$, we sample a random $\TT_i$ by sampling 
  each $\bT_i(k)$ independently and
uniformly (with replacement) from $[n]$ for each $k\in [n]$.
We will refer to $T_i$ as the $i$-th term in $T$ and 
  $T_i(k)$ as the $k$-th variable of $T_i$.
Given a term $T_i:[n]\rightarrow [n]$, we abuse the notation to 
  use $T_i$ to denote the Boolean function over $\{0,1\}^n$ with
  $T_i(x)=1$ if $x_{T_i(k)}=1$ for all $k\in [n]$ and $T_i(x)=0$ otherwise.
  (So $T_i$ is a conjunction, which is why we refer to it as a term as mentioned above.)

A function $\ff \sim \Dy$ is drawn using the following procedure:
\begin{flushleft}\begin{enumerate}
\item Sample $\TT\sim\calE$ and use it to define the multiplexer map $\bGamma = \bGamma_{\TT} \colon \{0, 1\}^n \to [L] \cup \{ 0^*, 1^* \}$:
$$
\bGamma_\TT(x) =\begin{cases}
0^* & T_i(x)=0\ \text{for all $i\in [L]$} \\
1^* & T_i(x)=1\ \text{for at least two different $i\in [L]$}\\
i & T_i(x)=1\ \text{for a unique $i\in [L]$}.
\end{cases}
$$ 
\item Sample  $\HH = (\hh_{i} \colon i \in [L])$ from $\Ey$, where each $\hh_{i} \colon \{0, 1\}^n \to \{0,1\}$ is independently (1) with probability $2/3$, a random dictatorship Boolean function, i.e., $\hh_{i}(x) = x_k$ with $k$ sampled  uniformly at random from $[n]$; and (2) with probability $1/3$, the all-$0$ 
    function.
\item Finally, with $\TT$ and $\HH$  in hand, $\ff =f_{\TT,\HH}\colon \{0, 1\}^n \to \{0, 1\}$ is defined as follows: $\ff(x) = 1$ if $\Vert x \Vert_1 > 3n/4 + 1$; $\ff(x) = 0$ if $\Vert x \Vert_1 < 3n/4 $; if $ \Vert x \Vert_1\in \{3n/4,3n/4+1\}$, we have
\[ 
\ff(x) = \begin{cases} 0 & \bGamma_{\TT}(x) = 0^* \\
						1 & \bGamma_{\TT}(x) = 1^* \\
						\hh_{\bGamma(x)}(x) & \text{otherwise (i.e., $\bGamma_{\TT}(x) \in [L]$).} \end{cases} 
\]
\end{enumerate}\end{flushleft}
A function $\ff \sim\Dn$ is drawn using the same procedure, with the exception that $\HH = (\hh_{i} \colon i \in [L])$ is drawn from $\En$ (instead of $\Ey$): Each $\hh_i:\{0,1\}^n\rightarrow \{0,1\}$ is (1) with probability $2/3$ a random anti-dictatorship Boolean function $\hh_i(x) = \overline{x_k}$ with $k$ drawn uniformly from $[n]$; and (2) with 
    probability $1/3$, the all-$1$ function.  Note that every function $f$ in the support of either $\Dy$ or $\Dn$ is a two-layer function as defined in \Cref{eq:two-layer}.

Our construction of $\Dyes$ and $\Dno$ differs in various ways from the construction for the two-sided, non-adaptive lower bound in \cite{CWX17stoc}, such as the number of terms and the biasing of the $\bh_i$ functions in $\Ey$ and $\En$.

In the next two lemmas we show that functions in the support of $\Dyes$ are monotone 
  and $\ff\sim\Dno$ is likely to have large relative distance from monotone functions.

\begin{lemma}\label{lem:mono}
Every function in the support of $\Dy$ is monotone. 
\end{lemma}

\begin{proof}
Fix a $T$ from the support of $\calE$ and $H$ from the support of $\Ey$. This fixes a function $f$ in the support of $\Dy$.

Note that for two-layer functions, the only nontrivial points lie on layers $\{3n/4,3n/4+1\}$. Thus, consider any $x\prec y$ such that $\Vert x \Vert_1=3n/4$ and $\Vert y\Vert_1=3n/4+1$. Note that they in fact only differ on one bit. Assume that $f(x)=1$, and we will show that $f(y)=1$. 

Since $f(x)=1$, we know that $x$ satisfies at least one term. Thus $y$ satisfies at least one term. If $y$ satisfies multiple terms, then $f(y)=1$. So consider the case where $y$ satisfies a unique term $T_{i'}$. Since $f(x)=1$, it must be the case that $x$ also uniquely satisfies the term $T_{i'}$, which implies $h_{i'}\neq 0$. Thus, $h_{i'}=x_k$ for some $k\in[n]$. Since $f(x)=1$, we have $x_k=1$, which implies $y_k=1$ and $f(y)=1$. 
This finishes the proof of the lemma.
\end{proof}

\begin{lemma}\label{lem:nonmono}
A function $\ff \sim \Dn$ satisfies $\reldist(\ff,\cmon)=\Omega(1)$ with probability $\Omega(1)$.  
\end{lemma}
\begin{proof}
We will show that with probability at least $0.0001$, the number of disjoint violating pairs for monotonicity is $\Omega(|\ff^{-1}(1)|)$. The lemma then follows directly from \Cref{lem:vio-edges-lb} and \Cref{remark:vio-edges}.

Fix a function $f$ from the support of $\Dn$ (equivalently, fix $T$ from the support of $\calE$ and $H$ from the support of $\En$.) Note that violating pairs can only appear at the two non-trivial layers $\{3n/4,3n/4+1\}$. Thus, at a high level, we want to show there are many disjoint pairs $x\prec y$ such that 1) $\Vert x \Vert_1=3n/4$ and $\Vert y \Vert_1=3n/4+1$, and 2) $f(x)=1$ and $f(y)=0$. 

Fix an $x\in\{0,1\}^n$ such that $\Vert x \Vert_1=3n/4$. Observe that if $x$ is involved in some violating pair, then it must be the case that $x$ satisfies a unique term $T_{i'}$ for $i'\in[L]$ and $h_{i'}=\overline{x_k}$ for some $k\in[n]$ such that $x_{k}=0$. So the conditions above are necessary. It is also easy to see if they hold, then $f(x)=1$, making $x$ a candidate to be part of some violating pair.

Let $y$ be the string that is obtained by flipping the $k$th bit of $x$. Ideally, we want to conclude that $f(y)=0$. Indeed, $h_{i'}(y)=0$, but we need to be careful to make sure that $y$ still satisfies the unique term $T_{i'}$. 

Formally, let $X$ be the set of all points $x$ such that 1) $\Vert x \Vert_1=3n/4$, 2) $x$ satisfies a unique term $T_{i'}$ for $i'\in[L]$, and $h_{i'}=\overline{x_k}$ for some $k\in[n]$ such that $x_{k}=0$, 3) $y$ (defined by flipping the $k$th bit of $x$) still satisfies the unique term $T_{i'}$.

Note that every $x \in X$ and the corresponding $y$ form a disjoint violating pair (since they must satisfy the unique term so that the anti-dictatorship function is well-defined), thus the number of disjoint violating pairs is lower bounded by the size of $X$. Note that $\bX$ is a random variable depending on the choices of $(\bT,\bC)$ and $\bH$. Next, we show that every fixed $x$ with $\Vert x \Vert_1=3n/4$ is in $\bX$ with constant probability.

\begin{claim} \label{claim:a}
  For each $x\in\{0,1\}^n$ such that $\Vert x \Vert_1=3n/4$, we have 
  
  \[
    \Prx_{\TT\sim\calE,\HH\in \En}[x\in \bX]\geq 0.01.
  \]
\end{claim}
\begin{proof}
We show the claim by calculating the probabilities in two steps. Fix $x\in\{0,1\}^n$ such that $\Vert x \Vert_1=3n/4$ and $k\in[n]$ such that $x_k=0$. Let $y$ be $x$ but with the $k$-th coordinate changed to 1. Note that $\Vert y \Vert_1=3n/4+1$.

First, the probability that $x$ and $y$ uniquely satisfy a term $\TT_i$ for some $i\in[L]$ is 
\[
L\cdot \left(\frac{3}{4}\right)^n\cdot \left(1-\left(\frac{3}{4}+\frac{1}{n}\right)^n\right)^{L-1}\geq 1/e-0.01\geq 0.35,
\]
for some sufficiently large $n$.

Fix such an $i'\in[L]$. The probability that $\hh_{i'}=\overline{x_k}$ is $\frac{2}{3n}$. Since $x$ has $n/4$ coordinates that are 0, we have
\[
\Prx_{\TT\sim\calE,\HH\in \En}[x\in \bX]\geq 0.35\cdot 1/6\geq 0.05.
\]
This finishes the proof of \Cref{claim:a}.
\end{proof}

Note that the number of points on the $3n/4$ layer is $${n\choose 3n/4}\geq \frac{|f^{-1}(1)|}{100}.$$ Thus, the expected size of $\bX$ is at least $|f^{-1}(1)|/10000$. Thus, we have $|\bX|=\Omega(|f^{-1}(1)|)$ with probability at least 0.0001.
This finishes the proof of \Cref{lem:nonmono}.
\end{proof}

Let $\eps_0$ and $c$ be the two constants hidden in the  $\Omega(1)$'s in 
  \Cref{lem:nonmono}, i.e., a function $\ff\sim \Dno$ has relative distance to monotonicity at 
  least $\eps_0$ with probability at least $c$.
Let $\alpha=\min(\eps_0,c)/100$.

We will consider
deterministic, $q$-query, 
  (two-stage) non-adaptive, $\MQ$-only algorithms $\ALG$,
  which are defined as follows.
  
  \begin{definition} \label{def:two-stage}
  A \emph{deterministic, $q$-query, 
  (two-stage) non-adaptive, $\MQ$-only algorithms $\ALG$} runs on $\Dyes$ and $\Dno$ as follows.  
Upon an input function $f$ from either $\Dyes$ or $\Dno$, 
\begin{flushleft}\begin{enumerate}
    \item In the first stage, ALG draws a sequence of $q$ samples $\bS$ independently and uniformly at random from the set of all points at or above layer $3n/4$, and queries all of them.
\item In the second stage, based on $\bS$ and the query results (denoted by $f(\bS)$,
  a $q$-bit string), ALG picks deterministically
  a set $\bQ=Q(\bS,f(\bS))$ of $q$ points to query,
  and finally accepts or rejects based on all the information in $(\bQ,\bS,f(\bQ\cup \bS))$.
\end{enumerate}\end{flushleft}
  \end{definition}

In order to prove \Cref{thm:two-sided-non-adaptive-lb-hehe}, it suffices to prove the following lemma:

\begin{lemma} \label{non-adaptive-mono-bound}
Let $q=n / (C_0\log n)$ for some sufficiently large constant $C_0$.
For any deterministic, $q$-query, (two-stage) non-adaptive, $\MQ$-only algorithm $\ALG$, we have 
\[ \mathop{\Pr}_{\ff\sim\Dy}\big[\ALG \text{ accepts }\ff\big] \leq (1+\alpha)\cdot \mathop{\Pr}_{\ff\sim\Dn}\big[\ALG\text{ accepts }\ff\big] + 2\alpha. \] 
\end{lemma}

To see that this proves the lower bound in \Cref{thm:two-sided-non-adaptive-lb-hehe}, 
  assume for a contradiction that there is a 
  randomized, $o(q)$-query, non-adaptive
  algorithm that accepts every monotone function with probability at least $1-\alpha$
  and rejects every function that has relative distance at least $\eps_0$ to monotonicity
  with probability at least $1-\alpha$.
Then by Yao's minimax principle, there
  exists a deterministic, $q$-query, (standard)
  non-adaptive algorithm $\ALG'$ that satisfies
$$
\mathop{\Pr}_{\ff\sim\Dy}\big[\ALG' \text{ accepts }\ff\big] - \mathop{\Pr}_{\ff\sim\Dn}\big[\ALG'\text{ accepts }\ff\big]\ge (1-\alpha)-\big(1-c(1-\alpha)\big) \ge c\big/2.
$$

Note that given any input function $f$ from either 
  $\Dyes$ or $\Dno$, such an algorithm $\ALG'$ works by first drawing a sequence of $o(q)$ samples from $f^{-1}(1)$ and then, based on the samples drawn,~picking deterministically a set of $o(q)$ points to query.
  Using \Cref{claim:two-layer1},  with a suitably small choice of the constant $\tau>0$, we can obtain from $\ALG'$ a deterministic, $o(q)$-query, (two-stage) non-adaptive, $\MQ$-only algorithm $\ALG$ that satisfies
$$
\mathop{\Pr}_{\ff\sim\Dy}\big[\ALG \text{ accepts }\ff\big] - \mathop{\Pr}_{\ff\sim\Dn}\big[\ALG \text{ accepts }\ff\big]\ge  c/4,
$$
which, however, contradicts with \Cref{non-adaptive-mono-bound}  using the choice of $\alpha$.

\subsection{Proof of \Cref{non-adaptive-mono-bound}}

Let $\ALG$ be a deterministic,
$q$-query, (two-stage) non-adaptive, MQ-only algorithm in the rest of the section, where the~parameter $q=n/(C_0 \log n)$ for some sufficiently large constant $C_0$. 
We start with a definition that gives the \emph{outcome} of an $f$ from either $\Dyes$ or $\Dno$ on a set $V$ of query points.

\begin{definition}
Let $f=f_{T,H}$ be a function from the support of either $\Dyes$
  or $\Dno$, and $V$  be a set of points that lie at 
  or above layer $3n/4$.
We
  define the following tuple $(P;R;\rho)$ as the \emph{outcome} of $f$ on $V$,
  where $P=(P_i:i\in [L])$ with $P_i\subseteq V$, $R=(R_i:i\in [L])$ with
  $R_i\subseteq V$, and $\rho=(\rho_i: i\in [L])$ with $\rho_i:P_i\rightarrow \{0,1\}$ for each $i\in [L]$:
\begin{flushleft}\begin{enumerate}
 	\item Start with $P_i=R_i=\emptyset$ for all $i\in [L]$ (and $\rho_i$
	  being the empty map given that $P_i=\emptyset$).
	\item We build $P$ and $R$ as follows. For each $x\in V$ that lies in layer $3n/4$ or $3n/4+1$ (the order does not matter), if no term in $T$ is satisfied, add $x\rightarrow R_i$ for 
		all $i\in [L]$; if $x$ satisfies a unique term in $T$, say $T_i$, 
		add $x\rightarrow P_i$ and $x\rightarrow R_j$
		for all $j\ne i$;
		if $x$ satisfies more than one term in $T$, letting $i<i'$ be the
		first two terms satisfied by $x$, add $x\rightarrow P_i$, $x\rightarrow P_{i'}$, and $x\rightarrow R_j$ for all $j: j<i'$ and $j\ne i$.
	\item Finally, for each $i\in [L]$ with $P_i\ne\emptyset$, set $\rho_i(x)=h_i(x)$ for each $x\in P_i$. 	
\end{enumerate}\end{flushleft}
Observe that having $x \in P_i$ indicates that $x$ satisfies term $T_i$ (though not every $x$ that satisfies a term $T_i$ will be in $P_i$), and that having $x \in R_i$ indicates that $x$ does not satisfy term $T_i$ (though not every $x$ that does not satisfy term $T_i$ will be in $R_i$).

We will denote the outcome $(P;R;\rho)$ of $f$ on $V$ as $\outc (f,V)$.
\end{definition}

We record the following facts that are evident by construction:

\begin{fact}\label{fact:1}
The $P$ and $R$ parts of the outcome depend on $T$ only.
Once $T$ is fixed, the $\rho$ part of the outcome depends on $H$ only.
Moreover, values of $f(x)$, $x\in V$, can be determined from $\outc(f,V)$.
\end{fact}

\begin{fact}\label{fact:3}
$\sum_{i\in [L]} |P_i|\le 2|V|$.
\end{fact}

\def\Hyes{\mathcal{H}_{\text{yes}}}
\def\Hno{\mathcal{H}_{\text{no}}}
\def\brho{\boldsymbol{\rho}}

Let $\Hyes$ denote the following distribution:
To draw a tuple $\bU=(\bQ,\bS,(\bP;\bR;\brho))\sim\Hyes$, we first draw $\ff\sim \Dyes$ and a sequence of $q$ points $\bS$ uniformly at random from points that are at or above layer $3n/4$.
$\bQ$ is then the set of $q$ points that $\ALG$ queries in the second stage,
  given $\bS$ and $\ff(\bS)$ in the first stage, and 
  $(\bP;\bR;\brho)$ is set to be $\outc(\ff,\bQ\cup \bS)$.
We define $\Hno$ similarly, with the function $\ff\sim \Dno$ instead of $\Dyes$.

\begin{definition}\label{def:bad-events}
Let $U=(Q,S,(P;R;\rho))$ be a tuple in the support of either $\Hyes$ or $\Hno$.
We say $U$ is \emph{bad} if at least one of the following 
  two events happen:
\begin{flushleft}\begin{itemize}
\item For some $i \in [L]$, there exist $x, y \in P_i$ such that $$\big|\{ k \in [n] :x_k = y_k = 1\}\big| \leq 3n/4 - 100 \log n.$$
\item For some $i \in [L]$, there exist $x,y\in P_i$ such that $\rho_i(x)\ne \rho_i(y)$.
\end{itemize}\end{flushleft}
Otherwise, we say $U$ is \emph{good}.
\end{definition}

\Cref{non-adaptive-mono-bound} follows from the following two claims:

\begin{claim}\label{lem:prune-non-adaptive}
A tuple $\bU\sim \Hyes$ is bad with probability at most $2\alpha $.
\end{claim}

\begin{claim}\label{lem:non-adaptive2}
For any fixed good tuple $U$ in the support of $\Hyes$, we have 
\[ \Prx_{\bU\sim \Hyes}\big[\hspace{0.03cm}\bU=U  \hspace{0.03cm}\big] \leq (1+\alpha)\cdot  \Prx_{\bU\sim \Hno}\big[\hspace{0.03cm}\bU=U\hspace{0.03cm}\big]. \]
\end{claim}

Before proving \Cref{lem:prune-non-adaptive} and \Cref{lem:non-adaptive2}, we use them to prove \Cref{non-adaptive-mono-bound}:

\begin{proof}[Proof of \Cref{non-adaptive-mono-bound} assuming  \Cref{lem:prune-non-adaptive} and \Cref{lem:non-adaptive2}]
Let $\calU$ be the set of those
  $U$ in the support of $\Hyes$ that lead $\ALG$ to accept.
Using \Cref{lem:prune-non-adaptive} and \Cref{lem:non-adaptive2}, we have 
\begin{align}
\nonumber
\Prx_{\ff\sim\Dyes}\big[\ALG\ \text{accepts}\ \ff\big]
\nonumber &=\sum_{U\in \calU} \Prx_{\bU\sim \Hyes} \big[\bU=U\big]\\[0.2ex]&\label{simpleeq1}\le \sum_{\text{good}\ 
U\in \calU} \Prx_{\bU\sim \Hyes}\big[\bU=U\big]
 +2\alpha\\
&\label{simpleeq2}\le (1+\alpha) \sum_{\text{good}\ U\in \calU} \Prx_{\bU\sim \Hno}\big[\bU=U\big] +2\alpha \\
&\nonumber \le (1+\alpha)\cdot \Prx_{\ff\sim \Dno}\big[\text{$\ALG$ accepts $\ff$}\big]+2\alpha,
\end{align}
where \Cref{simpleeq1} used \Cref{lem:prune-non-adaptive} and  \Cref{simpleeq2} used \Cref{lem:non-adaptive2}, 
and the last inequality is because
\[
\Prx_{\ff\sim \Dno}\big[\text{$\ALG$ accepts $\ff$}\big]
=
\sum_{\text{good}\ U\in \calU} \Prx_{\bU\sim \Hno}\big[\bU=U\big]
+
\sum_{\text{bad}\ U\in \calU} \Prx_{\bU\sim \Hno}\big[\bU=U\big].
\]

This finishes the proof of the lemma.
\end{proof}

\subsection{Proofs of Claims}\label{sec:claims}

Now we prove \Cref{lem:prune-non-adaptive} and \Cref{lem:non-adaptive2}. We start with \Cref{lem:prune-non-adaptive}.

We divide the proof of \Cref{lem:prune-non-adaptive} 
  into two steps.
First we introduce the following distribution
  $\Hyes^*$ for the analysis of the first stage of $\ALG$.
To draw a tuple $\bU^*=(\bS,(\bP^*;\bR^*;\brho^*))\sim \Hyes^*$, we first draw a function $\ff\sim \Dyes$ and a sequence of 
  $q$ points $\bS$ uniformly at random from points that are at or above layer $3n/4$.
The triple $(\bP^*;\bR^*;\brho^*)$ is then set to be $\outc(\ff,\bS)$.
\begin{definition} 
Let $U^*=(S,(P^*;R^*;\rho^*))$ be a tuple in the support of $\Hyes^*$. 
We say $U^*$ is \emph{regular} if it satisfies the following two conditions: (1)
$|P^*_i|\le 1$ for all $i\in [L]$, and (2) every two strings in $S$ have at most $5n/8$ many common $1$-entries.
\end{definition}
\begin{claim}\label{lemma:need1}
With probability at least $1-o_n(1)$,
$\bU^*=(\bS,(\bP^*;\bR^*;\brho^*))\sim \Hyes^*$ is regular.
\end{claim}
\begin{proof}[Proof of \Cref{lemma:need1}]
Let $\bx^1,\ldots,\bx^q$ be the sequence of samples in $\bS$ (drawn uniformly at random from layer $3n/4$ or above). 
It follows from Chernoff bound and a union bound that with probability at least $1-o_n(1)$, we have for every $1\le i<j\le q$, the number of $k\in [n]$ with $x^i_k=x^j_k=1$ is~at most $5n/8$ (note that the expectation is about $9n/16$).
Assuming this is the case for a fixed set of samples $x^1,\ldots,x^q$, it suffices to show that when $\bT_1,\ldots,\bT_L$ are drawn independently and uniformly at random, no $i,j,k$ satisfy $\bT_k(x^i)=\bT_k(x^j)$ with probability at least $1-o_n(1)$.

This is because by our assumption on $x^i$ and $x^j$, the probability of
$\bT_k(x^i)=\bT_k(x^j)=1$  for any triple $i,j,k$ is at most $
(5/8)^n
$. Given that $L=(4/3)^n$ and $q$ is at most polynomial in $n$, the claim follows from a union bound over all $Lq^2$ many triples $i,j,k$.
\end{proof}

Now we work on the second stage of $\ALG$.
Fix any regular $U^*=(S,(P^*;R^*;\rho^*))$ in the support 
  of $\Hyes^*$, and let us write $Q=Q_{U^*}$ to denote the set of $q$ queries $\ALG$ would query in the second stage when (1) the samples it queries are $S$ in the first stage and (2) the query results of $S$ are determined using $(P^*;R^*;\rho^*)$ as discussed in \Cref{fact:1}.
We write $\Dyes'$ to denote $\Dyes$ conditioning on $\ff\sim \Dyes$ satisfying $(P^*;R^*;\rho^*)=\outc(\ff,S)$,
  and $\Hyes'$ to denote the distribution of
  $\bU=(Q,S,(\bP;\bR;\brho))$ by first
  drawing $\ff\sim \Dyes'$ and then setting
  $(\bP;\bR;\brho)=\outc(\ff,Q\cup S)$.

We prove the following claim about $\Hyes'$ (after fixing   any regular $U^*$ in the support of $\Hyes^*$):

\begin{claim}\label{claim:lastone}
With probability at least $1-\alpha-o_n(1)$, $\bU\sim \Hyes'$ is good.
\end{claim}

 \Cref{lem:prune-non-adaptive}
  follows directly by combining \Cref{lemma:need1} and \Cref{claim:lastone}.

\begin{proof}[Proof of \Cref{claim:lastone}]
Fix a regular $U^* = (S, (P^*; R^*; \rho^*))$ in the support of $\Hyes^*$ and use it to define $Q,\Dyes'$ and $\Hyes'$ as above.
Let $(Q,S,(\bP;\bR;\brho))\sim \Hyes'$.

We start by showing that the first bad event, i.e.~there exist $x,y\in \bP_i$ for some $i\in [L]$ satisfying
\begin{equation}\label{heheeq}\big|\{ k \in [n]:x_k = y_k = 1\}\big| \leq (3n/4) - 100 \log n,\end{equation}
happens with probability at most $o_n(1)$. 
After proving this, we condition on it not happening and show that the second bad event, i.e., $\rho_i(x)\ne \rho_i(y)$ for some $i\in [L]$ and $x,y\in \bP_i$, happens with probability at most $\alpha$.
The claim then follows by a union bound.

For the first bad event, we start with the
  event that there exists some $i \in [L]$ such that $P_i^*=\emptyset$ (but $R^*_i$ needs  not to be empty in general) but some $x, y \in \bP_i$ satisfy \Cref{heheeq}.  
We prove this happens with $o_n(1)$ probability by a union bound over all $x,y\in Q$
(that satisfy \Cref{heheeq}) and all $i\in [L]$ with $P_i^*=\emptyset$.
For each such $x,y,i$, we note that $\bT_i$ (when we draw $\ff\sim \Dyes'$) is drawn uniformly at random from all terms $T_i$ that satisfy $T_i(z)=0$ for all $z\in R_i^*$.
So the probability of this event is 
$$
\frac{\Pr_{\bT_i}[\bT_i(z)=0\ \text{for all $z\in R^*_i$ and\ }\bT_i(x)=\bT_i(y)=1]}
{\Pr_{\bT_i}[\bT_i(z)=0\ \text{for all $z\in R^*_i$}]},
$$
where $\bT_i$ is drawn uniformly at random.
It is clear that the denominator is at least $1-o_n(1)$
  given that all $z\in R_i^*$ have weight at most $(3n/4)+1$.
The numerator, on the other hand, is at most
$$
\Prx_{\bT_i}\big[\bT_i(x)=\bT_i(y)=1\big] 
\le \left( \dfrac{(3n/4) - 100 \log n}{n} \right)^{n} = \left(\frac{3}{4}\right)^n\left(1-\frac{(4/3)100\log n}{n}\right)^n. $$
By a union bound over $q^2L$ triples of $x,y,i$, we have that this case happens with probability $o_n(1)$.

We next handle the event that there exists some $i\in [L]$ with $|P_i^*|=1$ such that some $x,y\in \bP_i$ satisfy \Cref{heheeq}.
By triangle inequality, a necessary condition for this event to happen is that there exists some $i\in [L]$ with $|P_i^*|=1$ (say $P_i^*=\{x^*\}$) such that some $y\in \bP_i$ satisfies 
\begin{equation}\label{heheeq2}
\big| \{ k \in [n]:x^*_k = y_k = 1\} \big| \leq (3n/4) - 50 \log n.
\end{equation}
Fix an $i\in [L]$ with $|P_i^*|=1$ (denoting $P_i^*=\{x^*\}$) and a $y\in Q$ such that \Cref{heheeq2} holds.
We bound the probability of $y\in \bP_i$ and then apply a union bound over the at most $2q^2$ many pairs of $i$ and $y$.

To this end, note that $\bT_i$ is drawn uniformly but subject to conditioning on $\bT_i(x^*)=1$ and $\bT_i(z)=0$ for all $z\in R^*_i$. So   the probability we are interested in is
$$
\frac{\Pr_{\bT_i}[\bT_i(x^*)=\bT_i(y)=1\ \text{and}\ \bT_i(z)=0\ \text{for all $z\in R_i^*$}]}
{\Pr_{\bT_i}[\bT_i(x^*)=1\ \text{and}\ \bT_i(z)=0\ \text{for all $z\in R_i^*$}]},
$$
where $\bT_i$ is drawn uniformly at random.
Note that the denominator is at least
$$
(3/4)^n\cdot \left(1-|R_i^*|\cdot\left(\frac{5/8}{3/4}\right)^n\right)\ge (1-o_n(1))\cdot (3/4)^n,
$$
where the $(3/4)^n$ is the probability of $\bT_i(x^*)=1$ and the rest is a lower bound for the probability of $\bT_i(z)=0$ for all $z\in R^*_i$ conditioning on $\bT_i(x^*)=1$, using the assumption that $x^*$ and $z$ share at most $5n/8$ common $1$-entries.
The numerator on the other hand is at most
$$
\Prx_{\bT_i}\big[\bT_i(x^*)=\bT_i(y)=1\big]
\le \left( \dfrac{(3n/4) - 50 \log n}{n} \right)^{n} = \left(\frac{3}{4}\right)^n \left(1-\frac{(4/3)50\log n}{n}\right)^n. 
$$
The probability of this case is at most $o_n(1)$ by taking the ratio and a union bound over the at most $2q^2$ many pairs of $i$ and $y$.

Assume the first bad event does not happen. Fix any such $T$, as well as $P$ and $R$ given $T$. 
Let $I:=\{i\in [L]:P_i\ne\emptyset\}$ and for each $i\in I$, let
	$$A_{i,1}:=\big\{k\in [n]: x_k=1\ \text{for all $x\in P_i$}\big\}\quad\text{and} \quad
A_{i,0}:=\big\{k\in [n]: x_k=0\ \text{for all $x\in P_i$}\big\}.$$
Using the assumption that the first bad event does not happen, we have:
\begin{flushleft}\begin{itemize}
\item For all $i \in [L]$ and $x, y \in P_i$, we have 
  $$\big|\{k \in [n]:x_k = y_k = 1\}\big| > 3n/4 - 100 \log n$$ 
  and thus,
  $$
  \big|\{k \in [n]:x_k=1 \cap y_k = 0\}\big|=\Vert x \Vert_1-\big|\{k \in [n]:x_k = y_k = 1\}\big|
  \le 100\log(n)+1 .
  $$
For any $i\in [L]$ with  any $x\in P_i$, we have 
$$
3n/4+1\ge \Vert x \Vert_1\ge |A_{i,1}|\ge \Vert x \Vert_1-\sum_{y\in P_i} \big|\{k\in [n]: x_k=1 \cap y_k=0\}\big|
\ge 3n/4-O( |P_i|\log n).
$$
Similarly, we have 
$$
n/4\ge n-\Vert x \Vert_1\ge |A_{i,0}|\ge n-\Vert x \Vert_1-\sum_{y\in P_i} \big|\{k\in [n]: x_k=0 \cap y_k=1\}\big|
\ge n/4-O(|P_i|\log n).
$$
\end{itemize}\end{flushleft}
Now recall that the second bad event is that $\rho_i(x)\ne \rho_i(y)$ for some $i\in [L]$ and $x,y\in \bP_i$.
For the second bad event to happen, letting  $I':=\{i\in [L]:P^*_i\ne \emptyset\}$ (so $I'\subseteq I$, $|I'|\le q$ and $|I|$ $\le 2q$), there exists either (1) an $i\in I'$ with $P_i^*=\{x^*\}$ such that $\brho_i(x^*)\ne \brho_i(y)$ for some $y\in P_i$, or (2) an $i\in I\setminus I'$ such that $\brho_i(x)\ne \brho_i(y)$ for some $x,y\in P_i$.
The probability of the former, for each $i\in I'$, is (given that we are in the yes case)
$$
1-\frac{|A_{i,1}|}{|j\in [n]: x_j^*=1|}\le O\left(\frac{|P_i|\log n}{n}\right)
$$
if $\rho_i^*(x^*)=1$ or
$$
1-\frac{|A_{i,0}|}{|j\in [n]: x_j^*=0|}\le O\left(\frac{|P_i|\log n}{n}\right)
$$
if $\rho_i^*(x^*)=0$.
The probability of the latter, for some $i\in I\setminus I'$, is at most
$$
1-\frac{|A_{i,0}|+|A_{i,1}|}{n}\le O\left(\frac{|P_i|\log n}{n}\right).
$$
So conditioning on  the first bad event not happening, the second bad event happens with probability at most
\[   \sum_{i \in I} O\left(\frac{|P_i|  \log n} {n}\right) \le \alpha, \]
by using $\sum_{i\in I} |P_i| \leq 2q = 2n / (C_0\log n)$ and 
  setting $C_0$ to be sufficiently large.
\end{proof}

Finally we prove \Cref{lem:non-adaptive2}:

\begin{proof}[Proof of \Cref{lem:non-adaptive2}]
Fix a good $U = (Q,S, (P ; R;\rho))$ in the support of $\Hyes$. Let $I:=\{i\in [L]:P_i\ne\emptyset\}$.
For each $i\in I$, let $A_{i,1}$ and $A_{i,0}$ be defined as in the proof of \Cref{claim:lastone}. Given that $U$ is good, $|A_{i,0}|$ and $|A_{i,1}|$ satisfy the same inequalities established in the proof of \Cref{claim:lastone}.

Consider any fixed $T$ such that the probability of $(P;R;\rho)=\outc(f_{T,\HH},Q\cup S)$  is positive when $\HH\sim\Ey$. Then it suffices to show that 
\[ \dfrac{\Pr_{\HH\sim\En}[\outc(f_{T \HH},Q\cup S) =(P;R;\rho)]}{\Pr_{\HH\sim\Ey}[\outc(f_{T ,\HH},Q\cup S)=(P;R;\rho) ]} \geq \frac{1}{1+\alpha}. \]
By definition of the distributions $\En$ and $\Ey$, this ratio is a product over $i\in I$. 

For each $i\in I$ with
 $\rho_i(x)=1$ for $x\in P_i$, the factor is 
\begin{equation}
    \frac{\frac{2}{3}\cdot \frac{|A_{i,0}|}{n}+\frac{1}{3}}{\frac{2}{3}\cdot \frac{|A_{i,1}|}{n}}=
    \frac{(2/3)\cdot |A_{i,0}|+(n/3)}{(2/3)\cdot |A_{i,1}|}.    \label{eq:ratio1}
\end{equation}
Given that $|A_{i,1}|=(3n/4)\pm O(|P_i|\log n)$ and $|A_{i,0}|=(n/4)\pm O(|P_i|\log n)$ (and $|P_i|\log n\le 2q\log n$ which can be set to be at most a sufficiently small constant multiple of $n$ by setting $C_0$ to be sufficiently small),  
\begin{equation}
\frac{(2/3)\cdot |A_{i,0}|+(n/3)}{(2/3)\cdot |A_{i,1}|} 
=\frac{(n/2)\pm O(|P_i|\log n)}{(n/2)\pm O(|P_i|\log n)}
=1\pm O\left(\frac{|P_i|\log n}{n}\right).\label{eq:ratio2}
\end{equation}
For each $i\in I$ with $\rho_i(x)=0$ for $x\in P_i$, the factor is
\begin{equation}
\frac{\frac{2}{3}\cdot \frac{|A_{i,1}|}{n} }{\frac{2}{3}\cdot \frac{|A_{i,0}|}{n}+\frac{1}{3}}=
\frac{(2/3)\cdot A_{i,1} }{(2/3)\cdot A_{i,0}+(n/3)}=1\pm O\left(\frac{|P_i|\log n}{n}\right).
\label{eq:ratio3}
\end{equation}
As a result, using $\sum_{i\in I} |P_i|\le 2q= 2n/(C_0\log n)$ and setting $C_0$ to 
  be sufficiently large, 
we have
\begin{align*}
\dfrac{\Pr_{\HH\sim\En}[\outc(f_{T, \HH},Q\cup S)=(P;R;\rho) ]}{\Pr_{\HH\sim\Ey}[\outc(f_{T, \HH},Q\cup S)=(P;R;\rho) ]}  
			  &\geq \prod_{i \in I} \left(1 - O\left( \dfrac{|P_i|\log n}{n} \right) \right) \ge \frac{1}{1+\alpha}.
\end{align*}
This finishes the proof of the claim.
\end{proof}

%% file: sections/two-sided-adaptive-lb.tex

\def\Eyes{\calE_{\text{yes}}}
\def\Eno{\calE_{\text{no}}}
\def\CC{\mathbf{C}}

\section{A two-sided adaptive lower bound}
\label{sec:two-sided-adaptive-lb}

In this section we prove the following two-sided adaptive lower bound:

\begin{theorem}~\label{thm:two-sided-adaptive-lb-hehe}
Let $\smash{N={n\choose 3n/4}}$. There is a constant $\eps_0>0$ such that any two-sided, adaptive algorithm for testing whether a function $f$ with $|f^{-1}(1)|=\Theta(N)$ is monotone or has
  relative~distance~at least $\eps_0$ from monotone functions must make $\tilde{\Omega}((\log N)^{2/3})$ queries.
\end{theorem}

Observations similar to those below
  the statement of \Cref{thm:two-sided-non-adaptive-lb-hehe} show that 
  \Cref{thm:two-sided-adaptive-lb-hehe} implies \Cref{main:adaptivelowerbound}.
As it will become clear in \Cref{sec:hehedist}, all functions used in our lower bound proof
  are two-layer functions as well.
It follows from \Cref{claim:two-layer1} that it suffices to show that 
  any randomized, adaptive, $\MQ$-only algorithm for testing
  two-layer functions with success probability $(99/100)\cdot 2/3$ must make $\tilde{\Omega}((\log n)^{2/3})$ queries.
So in the rest of the section, we focus on $\MQ$-only algorithms.

\subsection{Distributions $\Dyes$ and $\Dno$}\label{sec:hehedist}

Our proof is an adaptation of the $\tilde{\Omega}(n^{1/3})$ lower bound from \cite{CWX17stoc} (see Section~3) for adaptive monotonicity testing in the standard model.

We start with the distribution $\calE$. Let $L\coloneqq (4/3)^n$ and $M\coloneqq 4^n$.
$\calE$ is uniform over all pairs $(T, C)$ of the following
form: $T = (T_i:i\in [L])$ with $T_i:[n]\rightarrow [n]$ and  
$C = (C_{i,j} : i\in [L],j\in[M])$ with $C_{i,j} : [n]\rightarrow [n]$.
We call $T_i$'s the \emph{terms} and $C_{i,j}$'s the \emph{clauses}. 
Equivalently, to draw $(\TT ,\CC)\sim \calE$:
\begin{flushleft}\begin{itemize}
\item For each $i\in  [L]$, we sample a random term $T_i$ by sampling 
  $T_i(k)$ independently and uniformly from $[n]$ for each $k\in [n]$;
\item For each $i\in [L]$ and $j\in[M]$, we sample a random clause $C_{i,j}$ 
  by sampling $C_{i,j}(k)$ independently and uniformly from $[n]$ for each $k\in [n]$.
\end{itemize}\end{flushleft}

Given a pair $(T, C)$, we interpret $T_i$ as a term and 
  abuse the notation to write $$T_i(x)= \bigwedge_{k\in [n]} x_{T_i(k)}$$
as a Boolean function over $\{0,1\}^n$. We say $x$ satisfies $T_i$ when $T_i(x) = 1$. 
On the other hand, we interpret each $C_{i,j}$ as a clause 
  and abuse the notation to write $$C_{i,j}(x) = \bigvee_{k\in [n]} x_{C_{i,j}(k)}.$$
  
Each pair $(T,C)$ defines a multiplexer map $\Gamma = \Gamma_{T,C} : \{0, 1\}^n
  \rightarrow ([L]\times [M])\cup\{0^*, 1^*\}$ as follows: 
We have $\Gamma(x) = 0^*$ if $T_i(x) = 0$ for all $i\in [L]$ and 
  $\Gamma(x) = 1^*$ if $T_i(x) = 1$ for at least two different $i$'s in $[L]$. 
Otherwise there is a unique $i'\in [L]$ with $T_{i'}(x)=1$.
In this case the multiplexer enters the second level and examines 
  $C_{i',j}(x)$, $j\in [M]$. 
We have $\Gamma(x) = 1^*$ if $C_{i',j}(x)=1$ for all $j\in [M]$ and
  $\Gamma(x) = 0^*$ if $C_{i',j}(x)=0$ for at least two different $j$'s in $[M]$.
Otherwise there is a unique $j'\in [M]$ such that $C_{i',j'}(x)=0$,
  in which case the multiplexer outputs $\Gamma(x)=(i',j')$.
  
Next we define $\Dyes$ and $\Dno$.
A function $\ff\sim \Dyes$ is drawn using the following procedure:
\begin{flushleft}\begin{enumerate}
\item Sample a pair $(\TT,\CC)\sim \calE$, which is used to define a multiplexer map 
  $$\bGamma=\Gamma_{\TT,\CC}:\{0,1\}^n\rightarrow ([L]\times [M])\cup \{0^*,1^*\}.$$
\item Sample $\HH = (\bh_{i,j} : i\in [L], j\in[M])$ from $\Eyes$, where each $\hh_{i,j} : \{0, 1\}^n \rightarrow \{0, 1\}$ is 
  (1) with probability $2/3$, a random dictatorship Boolean function, i.e., $\hh_i (x) = x_k$ with $k$ sampled uniformly at random from $[n]$; and (2) with probability $1/3$, the all-$0$ function.
\item Finally, $\ff = f_{\TT ,\CC,\HH} : \{0, 1\}^n\rightarrow \{0,1\}$ is defined as follows:
  $\ff(x) = 1$ if $|x| > 3n/4+1$;
  $\ff(x) = 0$ if $|x| < 3n/4$; and 
  when $|x|\in \{3n/4,3n/4+1\}$, we have
$$
\ff(x)=\begin{cases}
0 & \text{if $\bGamma(x)=0^*$}\\
1& \text{if $\bGamma(x)=1^*$}\\
\hh_{\bGamma(x)}(x)& \text{otherwise ($\bGamma(x)\in [L]\times [M]$)}
\end{cases}
$$
\end{enumerate}\end{flushleft}
 
A function $\ff \sim\Dn$ is drawn using the same procedure, with the exception that $\HH = (\hh_{i,j} \colon i \in [L],j\in[M])$ is drawn from $\En$ (instead of $\Ey$): Each $\hh_i:\{0,1\}^n\rightarrow \{0,1\}$ is (1) with probability $2/3$ a random anti-dictatorship Boolean function $\hh_i(x) = \overline{x_k}$ with $k$ drawn and uniformly from $[n]$; and (2) with 
    probability $1/3$, the all-$1$ function.  
    
Our construction of $\Dyes$ and $\Dno$ again differs in various ways from the construction for the two-sided, adaptive lower bound in \cite{CWX17stoc}, such as the number of terms and clauses and the biasing of $\bh_i$ in $\Ey$ and $\En$. 
(See Figure~2 and Figure~3 of \cite{CWX17stoc} for some illustrations that
  can be helpful in understanding the $f_{T,C,H}$ functions.)   
    Note that every function in the support of either $\Dy$ or $\Dn$ is a two-layer function as defined in \Cref{eq:two-layer}.
   
   In the next two lemmas we show that functions in the support of $\Dyes$ are monotone and $\ff \sim\Dno$ is likely to have large relative distance from monotone functions.

\begin{lemma}
Every function $f$ in the support of $\Dyes$ is monotone.	
\end{lemma}
\begin{proof}
Fix a pair $(T,C)$ from the support of $\calE$ and $H$ from the support of $\Eyes$. This fixes a function $f$ in the support of $\Dyes$.

Note that for two-layer functions, the only nontrivial points lie on layers $\{3n/4,3n/4+1\}$. Thus, consider any $x\prec y$ such that $|x|=3n/4$ and $|y|=3n/4+1$. Note that they in fact only differ on one bit. Assume that $f(x)=1$, and we will show that $f(y)=1$. 

Since $f(x)=1$, we know that $x$ satisfies at least one term and falsifies at most one clause. This implies that $y$ satisfies at least one term and falsifies at most one clause as well. If $y$ satisfies multiple terms \textit{or} satisfies a unique term $T_{i'}$ and falsifies no clauses, then $f(y)=1$. So consider the case where $y$ satisfies a unique term $T_{i'}$ and falsifies a unique clause $C_{i',j'}$. Since $x$ is below $y$ and $f(x)=1$, we have $x$ uniquely satisfies the term $T_{i'}$ and uniquely falsifies the clause $C_{i',j'}$ as well. Again, since $f(x)=1$, we have $h_{i',j'}\neq 0$.

This means $h_{i',j'}=x_k$ for some $k\in[n]$. Since $f(x)=1$, we have $x_k=1$, which implies $y_k=1$ and $f(y)=1$. 
This finishes the proof.
\end{proof}

\begin{lemma}\label{lem:dist2}
A function $\ff \sim \Dno$ satisfies $\reldist(\ff,\cmon)=\Omega(1)$ with probability  $\Omega(1)$.  
\end{lemma}
\begin{proof}
We will show that with probability at least $0.0001$, the number of disjoint violating pairs to monotonicity is $\Omega(|\ff^{-1}(1)|)$. The lemma then follows directly from \Cref{lem:vio-edges-lb}.

Fix a function $f$ from the support of $\Dno$ (equivalently, fix $(T,C)$ from the support of $\calE$ and $H$ from the support of $\Eno$.) Note that violating pairs can only appear at the two non-trivial layers $\{3n/4,3n/4+1\}$. Thus, at a high level, we want to show there are many disjoint pairs $x\prec y$ such that 1) $|x|=3n/4$ and $|y|=3n/4+1$, and 2) $f(x)=1$ and $f(y)=0$. 

Fix an $x\in\{0,1\}^n$ such that $|x|=3n/4$. Observe that if $x$ is involved in some violating pair, then it must be the case that $x$ satisfies a unique term $T_{i'}$ for $i'\in[L]$ and falsifies a unique clause $C_{i',j'}$ for $j'\in [M]$. Furthermore, it needs to be that $h_{i',j'}=\overline{x_k}$ for some $k\in[n]$ such that $x_{k}=0$. So the conditions above are necessary. It's also easy to see if they hold, then $f(x)=1$, making $x$ a candidate point to be part of some violating pair.

Let $y$ be the string that is obtained by flipping the $k$th bit of $x$. Ideally, we want to conclude that $f(y)=0$. Indeed, $h_{i',j'}(y)=0$, but we need to be careful to make sure that $y$ still satisfies the unique term $T_{i'}$ and falsifies the unique clause $C_{i',j'}$. 

Formally, let $X$ be the set of points $x$ such that 1) $|x|=3n/4$, 2) $x$ satisfies a unique term $T_{i'}$ for $i'\in[L]$, falsifies a unique clause $C_{i',j'}$ for $j'\in [M]$, and $h_{i',j'}=\overline{x_k}$ for some $k\in[n]$ such that $x_{k}=0$, 3) $y$ (defined by flipping the $k$th bit of $x$ from 0 to 1) satisfies the unique term $T_{i'}$ and falsifies the unique clause $C_{i',j'}$.

Note that every $x$ and the corresponding $y$ form a disjoint violating pair (since they must satisfy the unique term and clause so that the anti-dictatorship function is well-defined), thus the number of disjoint violating pairs is lower bounded by the size of $X$. Note that $\bX$ is a random variable depending on the choices of $(\bT,\bC)$ and $\bH$. Next, we show that every fixed $x$ is in $\bX$ with constant probability:

\begin{claim}
  For each $x\in\{0,1\}^n$ such that $|x|=3n/4$, we have 
  
  \[
    \Prx_{(\TT,\CC)\sim\calE,\HH\in \Eno}[x\in \bX]\geq 0.01.
  \]
\end{claim}
\begin{proof}
We establish the claim by calculating the probability in three steps. Fix $x\in\{0,1\}^n$ such that $|x|=3n/4$ and $k\in[n]$ such that $x_k=0$. Let $y$ be $x$ but with the $k$th coordinate flipped to 1. Note that $|y|=3n/4+1$.

First, the probability that $x$ and $y$ uniquely satisfy a term $\TT_i$ for some $i\in[L]$ is 
\[
L\cdot \left(\frac{3}{4}\right)^n\cdot \left(1-\left(\frac{3}{4}+\frac{1}{n}\right)^n\right)^{L-1}\geq 1/e-0.01\geq 0.35,
\]
for some sufficiently large $n$.

Fix such an $i'\in[L]$. Second, the probability that $x$ and $y$ uniquely falsify a clause $\CC_{i',j}$ for some $j\in[M]$ is 
\[
M\cdot \left(\frac{1}{4}-\frac{1}{n}\right)^n\cdot \left(1-\left(\frac{1}{4}\right)^n\right)^{M-1}\geq 1/e-0.01\geq 0.35,
\]
for some sufficiently large $n$.

Fix such a $j'\in[M]$. The probability that $\hh_{i',j'}=\overline{x_k}$ is $\frac{2}{3n}$. Since $x$ has $n/4$ coordinates that are 0, we have
\[
\Prx_{(\TT,\CC)\sim\calE,\HH\in \Eno}[x\in \bX]\geq (0.35)^2\cdot 1/6\geq 0.02.
\]

This finishes the proof.
\end{proof}

Note that the number of points on the $3n/4$ layer is ${n\choose 3n/4}\geq |f^{-1}(1)|/100$. Thus, the expected size of $\bX$ is at least $|f^{-1}(1)|/10000$, and  $|\bX|=\Omega(|f^{-1}(1)|)$ with probability at least 0.0001.
\end{proof}

Let $\eps_0$ and $c$ be the two constants hidden in the $\Omega(1)$'s in \Cref{lem:dist2}, i.e., a function $\ff\sim\Dno$ has relative distance to monotonicity at least $\eps_0$ with probability at least $c$. Let $\alpha=\min(\eps_0, c)/100$
  and $q:=(n/C_0\log n)^{2/3}$ for some sufficiently large
  constant $C_0$ to be specified later.
Following arguments similar to those around \Cref{non-adaptive-mono-bound}
  in the previous section,
  it suffice to prove the following lemma for 
  deterministic, adaptive, $q$-query, $\MQ$-only algorithms.
    
\begin{lemma}\label{lem:kaka}
Let $\ALG$ be any deterministic, adaptive, $q$-query, $\MQ$-only algorithm. 
Then we have 
$$
\Prx_{(\TT,\CC)\sim \calE,\HH\sim \Eyes}
\left[ \text{$\ALG$ accepts $f_{\TT ,\CC,\HH}$}\right]
\le 
(1+\alpha)\cdot \Prx_{(\TT,\CC)\sim \calE,\HH\sim \Eno}
\left[ \text{$\ALG$ accepts $f_{\TT,\CC,\HH}$}\right]+ \alpha.
$$
\end{lemma}

\subsection{Proof of \Cref{lem:kaka}}

Now we start the proof of \Cref{lem:kaka}.
Formally, a deterministic, adaptive, $q$-query algorithm $\ALG$ is a 
  depth-$q$ binary decision tree, in which each internal node $u$
  is labelled with a query point $x_u$ and has two children with edges to them
  labelled with $0$ and $1$, respectively.
Each leaf of the tree is labelled either ``accept'' or ``reject.''
Given any function $f$, the execution of $\ALG$ on $f$ induces
 a path in the tree and $\ALG$ returns the label of the leaf at the end.
Since the goal of $\ALG$ is to distinguish $\Dyes$ from $\Dno$ (as described in 
the statement of \Cref{lem:kaka}), we may
  assume without loss of generality that every point queried in $\ALG$ 
  lies in the two layers of $3n/4$ and $3n/4+1$. 
  
Given a function $f=f_{T,C,H}$ from the support of either $\Dyes$ or $\Dno$,
  we define the \emph{outcome} of $\ALG$ on $f$ to be the following
  tuple $(Q;P;R;\rho)$, where $Q$ is the set of $\le q$ queries made by~$\ALG$ along the path,
 $P=(P_i,P_{i,j}:i\in [L],j\in [M])$, $R=(R_i,R_{i,j}:i\in [L],j\in [M])$ with every
  $P_i,P_{i,j}$, $R_i,R_{i,j}\subseteq Q$, and $\rho=(\rho_{i,j}:i\in [L],j\in [M])$ with $\rho_{i,j}:P_{i,j}\rightarrow \{0,1\}$. $P,R,\rho$ are built as follows:
\begin{flushleft}\begin{enumerate}
 	\item Start with $P_i=R_i=P_{i,j}=R_{i,j}=\emptyset$ for all $i\in [L]$ and $j\in [M]$.
	\item We build $P$ and $R$ as follows. For each query $x\in Q$ made along the path, if no term in $T$ is satisfied, add $x\rightarrow R_i$ for 
		all $i\in [L]$; 
		if $x$ satisfies more than one term in $T$, letting $i<i'$ be the
		first two terms satisfied by $x$, add $x\rightarrow P_i$, $x\rightarrow P_{i'}$, and $x\rightarrow R_k$ for all $k: k<i'$ and $k\ne i$;
		if $x$ satisfies a unique term in $T$, say $T_i$, 
		add  $x\rightarrow P_i$ and $x\rightarrow R_k$
		for all $k\ne i$.
For the last case we examine clauses $C_{i,j}$, $j\in [M]$, on $x$.
If no clause $C_{i,j}$, $j\in [M]$, is falsified by $x$, add $x\rightarrow R_{i,j}$ for all 
  $j\in [M]$;
if more than one clause $C_{i,j}$, $j\in [M]$, are falsified by $x$,
  letting $j<j'$ be the first two such clauses, then 
  add $x\rightarrow P_{i,j}$, $x\rightarrow P_{i,j'}$ and $x\rightarrow R_{i,k}$ for all
    $k:k<j'$ and $k\ne j$;
 if a unique clause $C_{i,j}$ is falsified by $x$ for some $j\in [M]$,
   add $x\rightarrow P_{i,j}$ and $x\rightarrow R_{i,k}$ for all $k\ne j$.

	\item For each $i\in [L],j\in [M]$ and $x\in P_{i,j}$, set $\rho_{i,j}(x)=h_{i,j}(x)$. 
	(So $\rho_{i,j}$ is the dummy empty map for any $i,j$ with $P_{i,j}=\emptyset$.)	
\end{enumerate}\end{flushleft}

We record the following facts about outcomes of $\ALG$:
  
\begin{fact}\label{fact:5}
For each $x\in Q$, the value of $f(x)$ is uniquely determined by the outcome of $\ALG$ on $f$. So
if the outcomes of $\ALG$ on two functions $f$ and $f'$ are the same, then
  $\ALG$ returns the same answer on all queries in $Q$.
\end{fact}  

\begin{fact}\label{fact:6}
For each $i\in [L]$, we have $\cup_{j\in [M]} P_{i,j}\subseteq P_i$.
We also have 
$$
\sum_{i\in [L]} |P_i|\le 2|Q| \quad\text{and}\quad \sum_{j\in [M]} |P_{i,j}|\le 2|P_i|,\quad\text{for every $i\in [L]$.}
$$
\end{fact}
  
We say the outcome $(Q;P;R;\rho)$ of $\ALG$ on $f$ (from the support of either $\Dyes$ or $\Dno$) is
  \emph{good} if it satisfies  the following two conditions:
\begin{flushleft}\begin{enumerate}
\item Condition $C_1$: For every $i\in [L]$ and $j\in [M]$ with $P_{i,j}\ne \emptyset$, letting
$$
A_{i,j,0}=\big\{ k\in [n]: x_{k}=1\ \text{for all $x\in P_{i,j}$}\big\}\quad\text{and}\quad
A_{i,j,1}=\big\{k\in [n]:x_k=0\ \text{for all $x\in P_{i,j}$}\big\},
$$
we have
\begin{align}\label{eq:333}
 \big||A_{i,j,0}|-n/4\big|,\hspace{0.1cm}\big||A_{i,j,1}|-3n/4\big| \le O\left(\min\big\{|P_{i,j}|^2,|P_i|\big\}\cdot \log n\right).
\end{align}
\item Condition $C_2$: For every $i,j$ with $P_{i,j}\ne \emptyset$, we have 
  $\rho_{i,j}(x)=\rho_{i,j}(y)$ for all $x,y\in P_{i,j}$.
\end{enumerate}\end{flushleft}
 
Similar to \Cref{non-adaptive-mono-bound} in the previous section,  
  \Cref{lem:kaka} follows from the next two claims:
 
\begin{claim}\label{lem:lolo}
Let $\ell=(Q;P;R;\rho)$ be a good outcome.
Then we have 
\begin{align*}
&\Prx_{(\TT ,\CC)\sim\calE,\HH\sim\Eyes}\big[\text{outcome of $\ALG$ on $f_{\TT,\CC,\HH}$ is $\ell$}\big]\\ 
&\hspace{4cm}\le (1 + \alpha) \cdot \Prx_{(\TT,\CC)\sim\calE,\HH\sim\Eno}\big[\text{outcome of $\ALG$ on $f_{\TT,\CC,\HH}$ is $\ell$}\big].
\end{align*}
\end{claim}
\begin{proof}
The proof is similar to that of \Cref{lem:non-adaptive2}.

Fix a good outcome $\ell=(Q;P;R;\rho)$.
Consider a fixed pair $(T,C)$ in the support of $\calE$ such that the probability of $(T,C, \HH)$ resulting in $\ell$ is positive when $\HH\sim\Ey$. Then it suffices to show that 
\begin{equation}\label{eq:444} \dfrac{\Pr_{\HH\sim\En}[(T,C, \HH) \text{ results in outcome $\ell$}]}{\Pr_{\HH\sim\Ey}[(T,C, \HH) \text{ results in outcome $\ell$}]} \geq \frac{1}{1+\alpha}. \end{equation}
Let $I$ be the set of $i\in [L]$ such that $P_i\ne \emptyset$;
for each $i\in [L]$, let $J_i$ be the set of $j\in [M]$ such that $P_{i,j}\ne \emptyset$.
Using (\ref{eq:333}) and an argument similar to that used in \Cref{lem:non-adaptive2},
  this ratio is at least
$$
1-\frac{1}{n}\cdot \sum_{i\in I}\sum_{j\in J_i} O\left(\min\big\{|P_{i,j}|^2,|P_i|\big\}\cdot \log n\right).
$$
As $\sum_{j\in J_i}|P_{i,j}|\le 2|P_i|$ by \Cref{fact:6}, we have
  that $$\sum_{j\in J_i} \min\big\{|P_{i,j}|^2,|P_i|\big\}$$ is maximized when
$|J_i|=2\sqrt{|P_i|}$ and $|P_{i,j}|=\sqrt{|P_i|}$, in which case
$$
\sum_{i\in I}\sum_{j\in J_i} \min\big\{|P_{i,j}|^2,|P_i|\big\}
=\sum_{i\in I} 2|P_i|^{3/2} \le O(q^{3/2})
$$
since $\sum_{i\in I} |P_i|\le 2q$.
(\ref{eq:444}) follows from the choice of $q$ and by setting $C_0$
  to be sufficiently large.
\end{proof}

\begin{claim}\label{lem:final}
We have 
\begin{align*}
 \Prx_{(\TT ,\CC)\sim\calE,\HH\sim\Eyes}\big[\text{outcome of $\ALG$ on $f_{\TT,\CC,\HH}$ is good}\big]\ge 1-\alpha.
\end{align*}
\end{claim}

We prove \Cref{lem:final} in \Cref{sec:lem:final}.
\Cref{lem:kaka} follows  from \Cref{lem:lolo} and \Cref{lem:final}, with a proof similar to that of \Cref{non-adaptive-mono-bound}.

\subsection{Proof of \Cref{lem:final}}\label{sec:lem:final}

First we describe the following equivalent view of
   running $\ALG$ on a function $f$ from the support of $\Dyes$:
After each query, $\ALG$ receives the outcome $(Q;P;R;\rho)$ 
  of queries made so far, where $Q$ is the set of queries made along the path 
  so far and 
  $P,R,\rho$ are built similar to before using $Q$.
From the outcome $(Q;P;R;\rho)$, $\ALG$ can uniquely recover 
  (\Cref{fact:5}) $f(x)$ for all queries made so far and for the most recent
  query in particular; it then uses the latter to walk down the tree to make the next query,
  and repeats this until a leaf is reached.  
  
With this view, we describe algorithm
  $\ALG'$ obtained by modifying $\ALG$ as follows:
After querying $x$, 
  the outcome $(Q;P;R;\rho)$ that $\ALG$ receives is updated differently as follows: 
\begin{flushleft}\begin{enumerate}
\item $x$ is added to $Q$, some sets of $P$ and some sets of $R$ in the same way as before.
\item When $x$ is added to a $P_{i,j}$ that is nonempty before $x$ is queried,
  instead of setting $\rho_{i,j}(x)$ to be $h_{i,j}(x)$, $\rho_{i,j}(x)$ is set to
  be the value of $\rho_{i,j}(y)$ of any $y\in P_{i,j}$ that was added to $P_{i,j}$ before $x$. (Note that if $x$ is the first point added to $P_{i,j}$, then $\rho_{i,j}(x)$ is set to be $h_{i,j}(x)$ as before.)
\end{enumerate}\end{flushleft}
Once $\ALG'$ receives the updated outcome, it recovers from it the value of $f$
  at $x$  
  (which could be different from the real value of $f(x)$) and simulates $\ALG$
  for one step to make the next query. 
Note that this modification ensures that at any time, the outcome of queries
  made so far as seen by $\ALG'$ always satisfies Condition $C_2$.
  
We prove the following two claims, from which \Cref{lem:final} follows.
We prove \Cref{finalclaim2} first, and then
  prove \Cref{finalclaim1} in the rest of the section.

\begin{claim}\label{finalclaim1}
The outcome of $\ALG'$ on $\ff\sim \Dyes$ satisfies Condition $C_1$ with probability
  $1-o(1)$.
\end{claim}
\begin{claim}\label{finalclaim2}
For any fixed outcome $\ell$ of $\ALG'$ on $\ff\sim \Dyes$ that satisfies Condition
  $C_1$, we have
$$
\Prx_{\ff\sim \Dyes}\big[\text{outcome of $\ALG$ on $\ff$ is $\ell$}\big] 
\ge (1-\alpha/2)\cdot \Prx_{\ff\sim \Dyes} \big[\text{outcome of $\ALG'$ on
  $\ff$ is $\ell$}\big].
$$
\end{claim}  
\begin{proof}
Fix an outcome $\ell=(Q;P;R;\rho)$ of $\ALG'$ that satisfies Condition $C_1$.
For each $P_{i,j}$ that is not empty, we let $x^{(i,j)}$ denote the 
  first point that is added to $P_{i,j}$.
Let $\rho_{i,j}=\rho_{i,j}(x^{(i,j)})$, which by the description of $\ALG'$ is the 
  same as $h_{i,j}(x^{(i,j)})$.
  
To prove the claim, we show that for any fixed $(T,C)$ such that 
  the probability that the outcome of $\ALG'$ on $\ff=f_{T,C,\HH}$ with $\HH\sim \Eyes$ is positive,
  we have
$$
\frac{\Pr_{\HH\sim \Dyes}[\text{outcome of $\ALG$ on $f_{T,C,\HH}$ is $\ell$}]}
{ \Pr_{\HH\sim \Dyes}[\text{outcome of $\ALG'$ on $f_{T,C,\HH}$ is $\ell$}]}\ge 1-\frac{\alpha}{2}.
$$  
Similar to arguments used before, this ratio can be written as a product of factors,
  one for each $i,j$ with $P_{i,j}\ne \emptyset$.
For any $i,j$ with $P_{i,j}\ne\emptyset$ and $\rho_{i,j}=0$, the factor is
$$
\frac{\frac{1}{3}+\frac{2}{3}\cdot \frac{|A_{i,j,0}|}{n}}
{\frac{1}{3}+\frac{2}{3}\cdot \frac{n-|x^{(i,j)}|}{n}}\ge 1-\frac{1}{n}\cdot O\big(\min\{|P_{i,j}|^2,|P_i|\}\log n\big).
$$
using the assumption that $\ell$ satisfies Condition $C_1$, where $A_{i,j,0}$ as before is the set of $k\in [n]$ with $x_k=0$ for all $x\in P_{i,j}$.
Similarly, when $\rho_{i,j}=1$ the factor is 
$$
\frac{\frac{2}{3}\cdot \frac{|A_{i,j,1}|}{n}}{\frac{2}{3}\cdot \frac{|x^{(i,j)}|}{n}}
\ge  1-\frac{1}{n}\cdot O\big(\min\{|P_{i,j}|^2,|P_i|\}\log n\big),
$$
again using Condition $C_1$. The rest of the proof follows from calculations similar to those at the end of the 
  proof of \Cref{lem:lolo}.
\end{proof}

To prove \Cref{finalclaim1},
  we start with a sufficient condition for the outcome of $\ALG'$ to satisfy Condition $C_1$.
We say the execution of $\ALG'$ on $f$ from the support of $\Dyes$ is \emph{regular}
  if the following condition holds in every round:
Let $(Q;P;R;\rho)$ be the outcome of queries revealed to $\ALG'$ so far
  (recall the modification of $\rho$ since we are interested in $\ALG'$ instead of $\ALG$), 
  let $x$ be the next query, and let $(Q';P';R';\rho')$ be the updated outcome after $x$ is queried.
Then they satisfy
\begin{itemize}
\item For every $i\in [L]$ with $P_i\ne \emptyset$,
$|A_{i,1}\setminus A_{i,1}'|\le 100\log n$, where
$$
A_{i,1}:=\big\{k\in [n]: y_k=1\ \text{for all $y\in P_i$}\big\}\quad\text{and}\quad
A_{i,1'}:=\big\{k\in [n]: y_k=1\ \text{for all $y\in P_i'$}\big\}.
$$
\item For every $i\in [L]$, $j\in [M]$ with $P_{i,j}\ne \emptyset$, 
  $|A_{i,j,0}\setminus A_{i,j,0}'|\le 100\log n$, where
$$
A_{i,j,0}:=\big\{k\in [n]: y_k=0\ \text{for all $y\in P_{i,j}$}\big\}\quad\text{and}\quad
A_{i,j,0}':=\big\{k\in [n]: y_k=0\ \text{for all $y\in P_{i,j}'$}\big\}
$$
\end{itemize}

\begin{claim}
If execution of $\ALG'$ on $f$ is regular, then
  the outcome of $\ALG'$ satisfies Condition $C_1$.
\end{claim}
\begin{proof}
Let $(Q;P;R;\rho)$ be the final outcome of $\ALG'$ on $f$. For each $i,j$ with $P_{i,j}\ne \emptyset$, let
$$
A_{i,1}:= \big\{k\in [n]: y_k=1\ \text{for all $y\in P_i$}\big\}\quad\text{and}\quad A_{i,j,0}:=\big\{k\in [n]: y_k=0\ \text{for all $y\in P_{i,j}$}\big\},
$$
with $A_{i,j,1}$ being defined similarly (using $y_k=1$ instead of $0$).

Given that the execution of $\ALG'$ on $f$ is regular, we have from the definition above that
$$
3n/4+1\ge |A_{i,1}| \ge (3n/4) -100 |P_i|\log n \quad\text{and}\quad
n/4\ge |A_{i,j,0}| \ge (n/4)-100 |P_{i,j}|\log n.
$$
So $A_{i,j,0}$ already satisfies \eqref{eq:333} given that $P_{i,j}\subseteq P_i$.
We focus on $A_{i,j,1}$ below.
For any $x, y\in P_{i,j}$, 
$$
\{k\in [n] : x_k = y_k = 0\}
\ge |A_{i,j,0}|\ge (n/4)-100 |P_{i,j}|\log n.
$$
Given that $|x|,|y|\in \{3n/4, 3n/4+1\}$, we have 
$$
\big|\{k\in [n] : x_k = 1, y_k = 0\}\big|\le 100|P_{i,j}|\log n.$$
As a result, we fix any $x\in P_i$ and have 
$$
|A_{i,j,1}|\ge |x|-\sum_{y\in P_{i,j}}\big|\{k\in [n]: x_k = 1, y_k = 0\}\big|
 \ge 3n/4- O\big(|P_{i,j}|^2 \log n\big).
$$
Combining this with $3n/4+1\ge |A_{i,j,1}|\ge |A_{i,1}|\ge (3n/4)-100 |P_i|\log n$
  finishes the proof.
\end{proof}

\begin{claim}
The execution of
 $\ALG'$ on $\ff\sim\Dyes$ is regular with probability at least $1-o(1)$. 
\end{claim}
\begin{proof}
Fix any outcome $(Q;P;R;\rho)$ that is revealed to $\ALG'$ after $k$ queries, $k\le q$,
  such that the execution so far is regular, and let $x$ be the next query. We show below that,
  conditioning on $\ff\sim \Dyes$ reaching $(Q;P;R;\rho)$, the probability that
  the execution of $\ALG'$ is no longer regular after querying $x$ is $o(1/q)$.
The lemma then follows by applying a union bound on $k$.

To this end, notice that by definition, the execution of $\ALG'$ remains 
  regular after querying $x$ unless either (1) there exists an
  $i\in [L]$ such that $P_i\ne \emptyset$, $x$ satisfies $\TT_i$ (and thus,
  could be added to $P_i$ potentially) but 
$$
\Delta_i:=\big\{k\in [n]: k\in A_{i,1}\ \text{but}\ x_k=0\big\}
$$
has size at least $100\log n$; or (2) there exist $i,j\in [L]$ such 
  that $P_{i,j}\ne \emptyset$, $x$ violates $\CC_{i,j}$ but 
$$
\Delta_{i,j}:=\big\{k\in [n]: k\in A_{i,j,0 }\ \text{but}\ x_k=1\big\}
$$
has size at least $100\log n$.
In the rest of the proof, we show that (1) happens with probability at most $o(1/q^2)$.
The same  upper bound for (2) follows from a symmetric argument. The lemma then follows by
  a union bound over the $O(q)$ many nonempty $P_i$ and $P_{i,j}$'s.
  
For (1) and any fixed $i$ with $P_i\ne \emptyset$,
  assuming that $|\Delta_i|\ge 100\log n$, 
  the conditional probability we are interested in can be written as $|U|/|V|$,
  where $V$ is the set of $T:[n]\rightarrow [n]$ such that 
$$
T(k)\in A_{i,1}\ \text{for all $k\in [n]$ and each $z\in R_i$ has $z_{T(k)} = 0$ for some 
$k\in [n]$ }
$$ 
and $V$ is the set of $T:[n]\rightarrow [n]$ such that
$$
T(k)\in A_{i,1}\setminus \Delta_i\ \text{for all $k\in [n]$ and each $z\in R_i$ has $z_{T(k)} = 0$ for some 
$k\in [n]$ }.
$$

Let $\ell= \log n$. First we write $U'$ to denote the following subset of
  $U$: $T\in U'$ if $T\in U$ and the number of $k \in [n]$ with
  $T(k)\in \Delta_i$ is $\ell$.
It suffices to show that $|V|/|U'|=o(1/q^2)$ given that $U'$ is a subset of $U$.
Next we define the following bipartite graph $G$ between
  $U'$ and $V$: $T'\in U'$ and $V\in V$ have an edge if and only if
  $T'(k)=T(k)$ for all $k\in [n]$ with $T'(k)\notin \Delta_i$.
From the construction, it is clear 
  each $T'\in U'$ has degree $|A_{i,1}\setminus \Delta_i|^\ell$.

To lowerbound the degree of a $T\in V$, note that one only needs at most $q$ 
  many variables of $T$
  to kill all strings in $R_i$. 
As a result, the degree of each $T\in V$ is at least 
$$
{n-q\choose \ell}
|\Delta_i|^\ell
$$
By counting the number of edges in $G$ in two different ways and using 
  $|A_{i,1}|\le n$ trivially,
we have
$$
\frac{|U'|}{|V|}\ge {n-q\choose \ell} \cdot \left(\frac{|\Delta_i|}{|A_{i,1}\setminus \Delta_i|}\right)^\ell \ge \left(\frac{n-q}{\ell}\cdot \frac{100\ell}{3/4n+1}\right)^\ell\gg q^2.
$$
 This finishes the proof of the lemma.
\end{proof}

%% file: sections/ap-standard-yes-relative-no.tex

\section{Separating standard  and relative-error testing}
\label{ap:standard-yes-relative-no}

We give a simple example of a property which provides a strong separation between the standard property testing model and the relative-error testing model.  The property is as follows:  for $n$ a multiple of three, let ${\cal C}$ be the class of all functions $g: \zo^n \to \zo$ for which $|g^{-1}(1)|$ is an integer multiple of $2^{2n/3}$. 

For any $\eps \geq 2^{-n/3}$, it is clear that every $f: \zo^n \to \zo$ has $\ham(f,g) \leq \eps$ for some function in ${\cal C}$, so for such $\eps$ there is a trivial 0-query standard-model $\eps$-testing algorithm (which always outputs ``accept''). On the other hand, fix any $\eps < 1$ and consider the following two distributions over functions from $\zo^n$ to $\zo$:

\begin{itemize}

\item ${\cal D}_\yes$:  A draw of $\boldf \sim {\cal D}_{\yes}$ is obtained by choosing a uniformly random function from the set of all functions that output 1 on exactly $2^{2n/3}$ inputs in $\zo^n$.

\item ${\cal D}_\no$:  A draw of $\boldf \sim {\cal D}_{\no}$ is obtained by choosing a uniformly random function from the set of all functions that output 1 on exactly ${\frac 1 2} 2^{2n/3}$ inputs in $\zo^n$.

\end{itemize}
It is clear that every function in the support of ${\cal D}_{\yes}$ belongs to ${\cal C}$, whereas every function in the support of ${\cal D}_{\no}$ has $\reldist(f,{\cal C}) =1 > \eps$.  We claim that any relative-error $\eps$-testing algorithm for ${\cal C}$ must make $\Omega(2^{n/3})$ black-box queries or samples.  To see this, first observe that any algorithm making $o(2^{n/3})$ black-box queries has only an $o(1)$ chance of every querying an input $x$  such that $\boldf(x)=1$, whether $\boldf$ is drawn from ${\cal D}_{\yes}$ or from ${\cal D}_{\no}$, and thus the two distributions are indistinguishable from $o(2^{n/3})$ queries.  But it is also the case that any algorithm making $m=o(2^{n/3})$ calls to $\SAMP(\boldf)$ will with probability $1-o(1)$ receive an identically distributed sample in both cases ($\boldf \sim {\cal D}_{\yes}$ versus $\boldf \sim {\cal D}_{\no}$), since in both cases the sample will consist of $m$ independent uniform distinct elements of $\zo^n$.

%% file: sections/ap-mono-needs-both.tex

\section{Both oracles are needed for relative-error monotonicity testing}
\label{ap:mono-needs-both}

\noindent {\bf A sampling oracle alone doesn't suffice:}  Let $A$ be any algorithm which uses only the sampling oracle to do relative-error monotonicity testing and makes $q$ calls to the sampling oracle, where $q$ is at most (roughly) $2^{n/2}$.  Consider the execution of $A$ on a function $f: \zo^n \to \zo$ which is either (1) the constant-1 function, or (2) a uniform random Boolean function.   The distribution of the sequence of samples received by $A$ will be almost identical in these two cases, since in both cases with high probability this will be a uniform random sequence of $q$ distinct elements of $\zo^n$.  Since in case (1) the function $f$ is monotone and in case (2) with overwhelmingly high probability $f$ has relative-distance at least 0.49 from every monotone function, it follows that such an algorithm $A$ cannot be a successful relative-error monotonicity tester.

\medskip

\noindent {\bf A black-box oracle alone doesn't suffice:}  Let $A$ be any algorithm which uses only a black-box oracle to do relative-error monotonicity testing and makes $q$ calls to the black-box oracle, where $q$ is at most $2^{n/2}$.  Consider the execution of $A$ on a function $f: \zo^n \to \zo$ which is either 
\begin{enumerate}
\item $f(x) = \Indicator[x_1 + \cdots + x_n \geq 9n/10]$, or
\item Let $\bS$ be a uniform random subset of $2^{0.47n}$ of the ${n \choose n/2}=\Theta(2^n/\sqrt{n})$ many strings in $\zo^n$ that have exactly $n/2$ ones.  The function $f$ outputs 1 on input $x$ if either $x_1 + \cdots + x_n \geq 9n/10$ or if $x \in \bS$.
\end{enumerate}

In case (2), a simple argument shows that any $2^{n/2}$-query algorithm such as $A$ will only have an $o(1)$ probability of querying an input (an element of $\bS$) on which case (1) and case (2) disagree, so with $1-o(1)$ probability the sequence of responses to queries that $A$ receives will be identical in these two cases.  Since in case (1) the function $f$ is monotone and in case (2) $f$ has relative-distance at least 0.49 from every monotone function, such an algorithm $A$ cannot be a successful relative-error monotonicity tester.